\documentclass[12pt]{article}
\usepackage[utf8]{inputenc}
\usepackage{amsfonts, amsthm, amsmath, graphicx, enumerate, verbatim, amssymb, mathrsfs,hyperref,array}
\usepackage[authoryear]{natbib}

\newtheorem{thrm}{Theorem}
\newtheorem{lemma}{Lemma}
\newtheorem{prop*}{Proposition}
\newtheorem{claim}{Claim}
\newtheorem{proposition}{Proposition}
\newtheorem{corollary}{Corollary}

\newtheorem{thm*}{Theorem}
\theoremstyle{definition}
\newtheorem{remark}{Remark}
\newtheorem{defin}{Definition}
\newtheorem{example}{Example}

\usepackage[top=1in, bottom=1in, left=1in, right=1in]{geometry}
\usepackage{booktabs}

\pdfoutput=1

\title{Distributions of Centrality on Networks\thanks{I am grateful to the editors and three anonymous referees, Drew Fudenberg, Kevin He, Maximilian Kasy, Scott Kominers, Jonathan Libgober, Eric Maskin, Aureo de Paula, Matthew Rabin, Elie Tamer, Maria Voronina, Muhamet Yildiz and especially Benjamin Golub and Matthew Jackson for useful conversations and comments. Declaration of interest: none.}}
 
\author{Krishna Dasaratha\thanks{Department of Economics, Harvard University. Email: krishnadasaratha@gmail.com}}
\date{\today}

\begin{document}

\maketitle

\begin{center}\textbf{Abstract}
\end{center}

We provide a framework for determining the centralities of agents in a broad family of random networks. Current understanding of network centrality is largely restricted to deterministic settings, but practitioners frequently use random network models to accommodate data limitations or prove asymptotic results.  Our main theorems show that on large random networks, centrality measures are close to their expected values with high probability. We illustrate the economic consequences of these results by presenting three applications: (1) In network formation models based on community structure (called stochastic block models), we show network segregation and differences in community size produce inequality. Benefits from peer effects tend to accrue disproportionately to bigger and better-connected communities. (2) When link probabilities depend on spatial structure, we can compute and compare the centralities of agents in different locations. (3) In models where connections depend on several independent characteristics, we give a formula that determines centralities `characteristic-by-characteristic'. The basic techniques from these applications, which use the main theorems to reduce questions about random networks to deterministic calculations, extend to many network games.


\pagenumbering{gobble}
\newpage

\pagenumbering{arabic}
\setcounter{page}{1}

%TODO: R3 comment on Prop 5

\section{Introduction}

In many settings of economic interest, agents benefit from connections to others. These peer effects depend on network structures, and better positioned agents can benefit much more than their less central counterparts. In education, for example, students form networks of friends and these connections affect academic achievement through group study, as a source of motivation, etc. Empirical evidence suggests that the impact of these peer effects on outcomes is approximated well by  measures of network centrality such as Katz-Bonacich centrality (\cite*{Calvo09}, \cite*{hahn2015teams}). More generally, measures of centrality and related quantities are crucial to understanding economic models from peer effects and quadratic games on networks to social learning models such as DeGroot updating.\footnote{For peer effects and network games beyond education, see among others \cite*{Ballester06}, \cite*{bramoulle2014strategic} and \cite*{Konig14}. For the role of eigenvector centrality in DeGroot learning, see \cite*{demarzo2003persuasion} and \cite*{golub10}.}

While there is a large literature on Katz-Bonacich and other centrality measures in deterministic settings, relatively little is known about centrality measures on stochastic networks. But in many applied settings, precise data about the full network is not available. Researchers instead use statistical models of network formation where links form with probabilities depending on agent characteristics. As a very simple example, one could model the social network in a school with black and white students by assuming two students of the same race are friends with probability $50\%$  while two students of different races are friends with probability $25\%$ (where all connections form independently). Moreover in theoretical work, varying parameters in models of random network formation often provides more insight than comparing particular deterministic networks.

%To give a simple example of such a random network model, suppose we wish to study a school where $\frac34$ of students are white and the remainder are black. We might assume that any two students of the same race are friends with probability $\frac14$, while two students of different races are friends with probability $\frac18$. Then we can compute how effort choices or outcomes depend on these link probabilities.

The current paper gives a framework for determining how central each agent in a large random network will be. With this framework, we can reduce questions about values and comparative statics of centrality measures to the better understood deterministic setting. For applied work, these formulas also provide justification for approximating agents' centralities based only on information about frequencies of various types of links in the absence of more detailed network data.

The two main theorems characterize the centralities of agents in a large family of random network formation models including stochastic block models, which are a basic and widely-used class of models allowing group structure. We focus on two common measures of centrality, eigenvector and Katz-Bonacich centrality, which measure how many neighbors an agent has with more central neighbors weighted more heavily. In the spirit of the law of large numbers, the theorems show that asymptotically with high probability all agents' centralities are close to values which we can compute from link formation probabilities.  The proofs of these theorems rely on random graph theory, and in particular utilize a recent result by \cite*{Chung11}.

We need several conditions to ensure centrality measures converge, and a key requirement is that the network is not too sparse. Whether a link between two particular agents forms is random, but if these agents have enough connections then this link realization only has a small effect on their network position. In addition to the main theorems, we discuss how centrality measures need not converge when the network is too sparse or unbalanced and give several examples.

Our main application of these theorems relates homophily (the tendency for links to form within communities more than between communities) and inequality of outcomes. Returning to the example of education, within-school homophily is an important factor in understanding how social interactions matter for educational outcomes (\cite*{Echenique06}). Most student populations at schools in the United States include groups of students of multiple races, and social connections tend to be denser within these groups than between groups. We examine the consequences of this segregation for overall distributions of performance and describe mechanisms by which network structures create or exacerbate inequality. 

In more detail, we ask how distributions of centrality change as we vary link probabilities to increase or decrease homophily. This approach assumes that outcomes depend linearly on centrality, which is the case in standard models such as \cite*{Calvo09} but need not hold in arbitrary non-linear models. To take these comparative statics we consider stochastic block models of network formation: agents are divided into several groups and the probability two agents are connected depends on whether they are in the same group. We compare distributions using a strong notion of relative inequality of outcomes, Lorenz dominance. By this measure, more segregated networks are indeed more unequal. So in our example of a school with black and white students, educational achievement would be more equal if the probability two students of the same race are friends decreased to $45\%$ or the chance of a friendship between races increased to $30\%$. In the context of \cite*{Calvo09}, the implication is that policy changes decreasing segregation within schools would decrease not only the racial achievement gap but also the overall achievement gap. We can also ask which groups benefit most in absolute terms from new links. Similar dynamics tends to persist, but we find a notable exception. When indirect connections are sufficiently valuable, adding connections between different groups actually benefits a well-connected majority group more than a disadvantaged minority.

We also discuss how our theorems could be applied to more general network formation models. An alternative to group structure is spatial structure, with agents situated in a continuous space and closer agents connecting with higher probabilities. Determining which locations are most advantageous is subtle: we give a numerical example where certain agents are relatively central when connections are very concentrated locally or there are many connections at long distances, but not for intermediate networks. Without our asymptotic results, this type of comparison would only be tractable via simulations. We also give a formula for centralities in a network depending on several independent characteristics, such as race or gender and geography. Computing the centrality of an agent reduces to taking the product of her centralities in separate networks each depending on only one characteristic.

While we focus on centrality, the methods introduced have broader implications for network models with linear structure, from quadratic games to learning processes. Modifications of our theorems apply to a number of other economically relevant quantities related to eigenvectors or powers of a network's adjacency matrix. One example is influence in the DeGroot learning model, which is the eigenvector with eigenvalue $1$ of a stochastic matrix derived from the adjacency matrix (\cite*{golub10}). Thus, we can relate influence to social groups or geographic locations. Another example is the social segregation index of \cite*{Echenique07}, so our methods also describe how individuals' segregation indices depend on group structure.

The results in this paper give a framework to analyzing distributions of centrality in two parts. First, our main theorems reduce characterizing these distributions asymptotically to a deterministic calculation involving the matrix of link probabilities. Second, the latter half of the paper carries out this in deterministic calculation in several applications. Jointly, these analyses let us relate centralities to parameters in random network models capturing segregation and geography.

In more detail, the structure is as follows: the remainder of this section discusses related literature. In Section~\ref{Model}, we introduce notation and describe our family of network formation models. Section~\ref{Centralities} defines and discusses eigenvector centrality and Katz-Bonacich centrality. The two main theorems about centrality measures are given in Section~\ref{thms}. Section~\ref{ineq} examines the impact of homophily on inequality in networks using stochastic block models. Other network formation models are discussed in Section~\ref{gen}. Section~\ref{conc} concludes, and proofs and further extensions appear in the Appendix.

\subsection{Related Literature}

A large literature studies how network structure matters for quadratic network games, and we provide techniques to extend these analyses to random networks. As observed by \cite*{Ballester06} as well as many subsequent papers, Katz-Bonacich centralities are equal or closely related to Nash equilibrium strategies in games with appropriate quadratic utility functions. These models are supported by a number of empirical papers identifying agent decisions and/or outcomes in areas such as education (\cite*{Calvo09}) and R\&D (\cite*{Konig14})  with Katz-Bonacich centrality. While certain properties of networks can be understood in a deterministic analysis, important features of network structure are best captured by random networks. In particular, Theorems \ref{thm1} and \ref{thm2} facilitate comparative statics on equilibria with respect to segregation, geographic clustering and network density.

Beyond the application to quadratic games, the current paper contributes to several active areas of research in network economics: literatures on centrality measures, inequality in networks and homophily. 

There are long-lived literatures in sociology and more recently economics and engineering aiming to quantify how central individuals are in a network (e.g. \cite*{Katz53} and \cite*{Bonacich87}). By exploring the mean field theory and comparative statics of eigenvector and Katz-Bonacich centrality, we add to a theoretical literature on centrality measures. One existing approach is to understand centrality measures in terms of their formal properties, as in the axiomatic characterizations by \cite*{Dequiedt17} and \cite*{Bloch17}. Rather than axiomatizing centrality measures, we describe how to calculate these measures and give comparative statics.

Most closely related to the current paper, simultaneous works by  \cite*{avella2017centrality} and \cite*{parise2018graphon} study asymptotic convergence of centrality measures and equilibria of network games as part of an analysis of graphons. Their focus is on characterizing centrality measures for graphons, which are a generalization of networks to settings with a continuum of agents that includes our large-population model as a special case. We focus on exploring the economic consequences of asymptotic convergence in large finite networks. In particular, by assuming that links grow at a faster rate, we obtain a sharper bound on the distance between vectors of centralities. This allows characterizing centralities of individual agents and thus taking comparative statics as in Section~\ref{ineq}.\footnote{Our large enough eigenvalues condition will be more restrictive than the model in \cite*{avella2017centrality}, which requires that the maximum expected degree is at least $O(\log n)$. This lets us bound the Euclidean distance between eigenvectors centralities by a constant and Katz-Bonacich centralities by a term of order $\sqrt{n}$, while Theorem 2 of \cite*{avella2017centrality} includes an additional factor of $\sqrt{\log{n}}$.}

%Most closely related to the current paper, simultaneous works by  \cite*{avella2017centrality} and \cite*{parise2018graphon} study asymptotic convergence of centrality measures and equilibria of network games as part of an analysis of graphons. Our results differ in two ways important to analyzing the economic consequences of this convergence. First, we provide a sharper bound on the distance between vectors of centralities, which allows characterizing centralities of individual agents and thus taking comparative statics as in Section~\ref{ineq}.\footnote{We bound the Euclidean distance between vectors of centralities by a constant, while Theorem 2 of \cite*{avella2017centrality} gives a bound of order $\sqrt{\log{n}}$.} Second, we allow some sparsity in networks rather than requiring agents' degrees to be proportional to the network size.\footnote{More precisely, Theorem~\ref{thm1}, Theorem ~\ref{thm2} and the applications in Section~\ref{ineq} require the weaker large enough eigenvalues condition.}

The applications in Section~\ref{ineq} connect to network economics literatures on inequality and homophily. There has been growing interest in how network structure affects inequality, though this research is diverse in topics and models. \cite*{Kets11} consider allocations which are stable with respect to deviations by highly connected subgroups, and show that denser networks lead to more equitable distributions. Their analysis uses the Lorenz dominance relation to compare allocations, and we use the same relation in Section~\ref{two}.  \cite*{Calvo04} and \cite*{Calvo07} study inequality in the context of job search using a model where employment information spreads through networks stochastically. Several papers on strategic network formation consider settings where central agents obtain disproportionate rents from their network position (\cite*{Goyal07} and \cite*{Hojman08}). We contribute to this varied literature by looking at a setting where network structure influences outcomes because of strategic complementarities and by examining how homophily matters for inequality.

Like centrality, homophily in networks has been an active research area in sociology, economics and computer science for decades (see \cite*{Mcpherson2001} for a survey). The methods in the current work are closest to those in \cite*{Golub12}'s work on homophily and the speed of learning. Like Golub and Jackson we consider stochastic block models, and we also use techniques from random matrix theory to reduce questions about the spectra of random networks to questions about a fixed deterministic network.\footnote{The relevant spectral quantity is the second eigenvalue in \cite*{Golub12} and the first eigenvector here.} Stochastic block models are also used in much of the community  detection literature, which studies algorithms for finding subgroups in homophilous networks (e.g. \cite*{Karrer11}). 

\section{Model}\label{Model}

In this section, we specify notation for networks and define a stochastic model of network formation.

\subsection{Notation}

A network is a set of nodes $N=\{1,...,n\}$ and a set of edges contained in $N\times N$. All networks will be undirected, so that $(i,j)$ is an edge whenever $(j,i)$ is an edge. The neighbors $N_i$ of node $i$ are the set of nodes connected to node $i$ by an edge.

A network is determined by its adjacency matrix $A$, which is defined by $A_{ij} = 1$ if there is an edge between agents $i$ and $j$ and $A_{ij} = 0$ otherwise.

A walk on a network is a finite sequence of vertices such that each pair of consecutive vertices in the sequence are connected by an edge. A walk containing $k+1$ vertices has length $k$.

Given a vector $\textbf{x} \in \textbf{R}^n$, the Euclidean norm is denoted by $\|\textbf{x}\|_2$. Given an $n \times n$ matrix $A$, the matrix $2$-norm $\|A\|_2$ is defined by $\sup_{\|\textbf{x}\|_2 = 1} \|A\textbf{x}\|_2.$ When $A$ is a symmetric matrix, $\|A\|_2$ is equal to the maximum absolute value of an eigenvalue of $A$.

\subsection{Random Networks}

We define random networks by generating links independently with link probabilities specified by a matrix $\bar{A}$. Each edge between agents $i$ and $j$ is formed with probability $\bar{A}_{ij}$. These links are generated independently, so that the entries of the adjacency matrix of the network are independent random variables.\footnote{This assumption is relaxed in Appendix~\ref{clustering}.}

Given an $n \times n$ matrix of link probabilities $\bar{A}$, we generate one instance of a network with $n$ agents and let $A$ be the adjacency matrix of this network.

Stochastic block models are our leading example. Consider a set of $n$ nodes divided into $m$ groups. In a stochastic block model, the probability of an edge between two agents depends only on their groups, so that $\bar{A}_{kl} = p_{ij}$ whenever agent $k$ is in group $i$ and agent $l$ is in group $j$. Because the amount of homophily in the network depends on parameters determining the sizes of groups and probabilities of links between groups, we can vary homophily in different ways by changing these parameters.

Note that both $\bar{A}$ and $A$ are symmetric matrices. As a consequence, these matrices have $n$ eigenvalues (counted with multiplicity), and these eigenvalues are real. Let $\overline{\lambda}_1, \overline{\lambda}_2,...,\overline{\lambda}_n$ be the eigenvalues of $\bar{A}$, ordered so that
$$|\overline{\lambda}_1| \geq |\overline{\lambda}_2| \geq ... \geq |\overline{\lambda}_n| .$$
Similarly, let $\lambda_1,\lambda_2,...,\lambda_n$ be the eigenvalues of $A$, ordered so that
$$|\lambda_1| \geq|\lambda_2| \geq ... \geq |\lambda_n|.$$

We assume that all link formation probabilities are positive, so that $\bar{A}$ is a positive matrix and satisfies the conclusions of the Perron-Frobenius theorem. In particular this matrix has a unique eigenvector with largest eigenvalue, $\overline{\lambda}_1$, and this eigenvector has non-negative real entries.

Finally, we consider sequences of random networks indexed by the population size $n$.

\begin{defin}
A \textbf{sequence of random networks} is, for each $n$ in a sequence of positive integers converging to infinity, an $n \times n$ matrix of probabilities $\bar{A}(n)$ and an adjacency matrix $A(n)$ of a network generated with these probabilities.
\end{defin}

\section{Centrality Measures}\label{Centralities}

We consider two common measures of centrality, which we refer to as Katz-Bonacich centrality and eigenvector centrality. These notions of centralities are defined by linear equations, and this linearity makes analyzing these centrality measures on random networks tractable. In subsequent sections, we discuss how the distributions of centralities of agents in random networks depend on parameters in the network formation model.

\subsection{Katz-Bonacich Centrality}

We first discuss the Katz-Bonacich centralities of agents in a network. This measure can be interpreted in terms of both linear algebra and graph theory.

Fix a positive constant $\phi < \|A\|_2^{-1}$. Recall $\|A\|_2$ is equal to the maximum of the norm $\|A\textbf{x}\|_2$ for a unit vector $\textbf{x}$, so that the condition implies $\|\phi A \textbf{x}\|_2$ is not too large relative to $\|\textbf{x}\|_2$.

\begin{defin}\label{katzbondef} The \textbf{Katz-Bonacich centrality} of agent $i$ with respect to the constant $\phi$ is given by $c_i(A,\phi)$, where $\textbf{c}(A,\phi)$ is the solution to
$$\textbf{c}(A,\phi) = \phi A \textbf{c}(A,\phi) + \mathbf{1},$$
where $\mathbf{1}$ is the column vector with all entries $1.$
\end{defin}

This equation is solved uniquely by
\begin{align*}\textbf{c}(A,\phi) &= (I - \phi A)^{-1} \textbf{1}\\
&= \sum_{k=0}^{\infty} \phi^k A^k \textbf{1}.\end{align*}

Because $(A^k)_{ij}$ is the number of walks from agent $i$ to agent $j$, this series has a combinatorial interpretation. The Katz-Bonacich centrality of agent $i$ is the number of walks beginning at agent $i$, with each walk discounted according to its length with discount factor $\phi$. So larger values of $\phi$ correspond to counting longer connections more heavily, while smaller values of $\phi$ correspond to counting shorter walks more.

A growing literature relates Katz-Bonacich centrality to strategies and outcomes in games on networks. Suppose that agents choose a level of effort which determines outcomes, and that the utility of the agent when effort levels are $\textbf{e}$ is $$u_i(\textbf{e}) = e_i - \frac12 e_i^2 + \phi \sum_{j=1}^n A_{ij} e_ie_j.$$This functional form implies that utility is quadratic in effort levels and that there are complementarities between an agent's effort and her neighbors' effort. Then at a Nash equilibrium the strategies $\textbf{e}^*$ satisfy
$$\textbf{e}^* = \phi  A \textbf{e}^* + \textbf{1}.$$
So the equilibrium strategies are given by Katz-Bonacich centralities:
$$\textbf{e}^* = \textbf{c}(A,\phi).$$
Therefore, in these games, effort is equal to Katz-Bonacich centrality and utility is a quadratic function of Katz-Bonacich centrality. If we assume that effort determines outcomes and that peer effects manifest through the costs to effort, then an agent's outcome is just her Katz-Bonacich centrality.

\subsection{Eigenvector Centrality}\label{EigCent}

Eigenvector centrality is a closely related notion. The eigenvector centrality of agent $i$ is the $i^{th}$ coordinate of the eigenvector of the adjacency matrix with eigenvalue of largest absolute value. More formally:

\begin{defin} The \textbf{eigenvector centrality} of agent $i$ in the network with adjacency matrix $A$ is given by $v_i(A)$, where $\textbf{v}(A)$ is the eigenvector of $A$ with largest eigenvalue, i.e. the solution to $$A\textbf{v}(A) = \lambda_1 \textbf{v}(A)$$
which satisfies $\|\textbf{v}(A)\|_2 = 1$.
\end{defin}

\begin{remark}
If the first eigenvalue has multiplicity greater than one, the eigenvector centrality may not be uniquely defined.\footnote{We assumed in Section~\ref{Model} that $\bar{A}$ is positive, which implies by the Perron-Frobenius theorem that $v(\bar{A})$ is unique. But even when $\bar{A}$ is positive, there is a positive probability that the realized matrix $A$ is not positive, so we cannot rule out the possibility of multiple eigenvectors corresponding to the eigenvalue $\lambda_1$.} Our results will not depend on the choice of convention in these cases.
\end{remark}

\begin{remark}Because we must choose a normalization, the ratios between the eigenvector centralities of different agents are more meaningful than the levels of each agent's centrality.
\end{remark}
An agent's eigenvector centrality is proportional to the sum of her neighbors' eigenvector centralities:
$$v_i(A) = \lambda^{-1} \sum_{j \in N_i}  v_j(A).$$
In other words, eigenvector centrality is a measure of how many neighbors an agent has, with more central neighbors counting more.

Eigenvector centrality can be thought of as a limit of Katz-Bonacich centralities. More precisely, the Katz-Bonacich centrality $\textbf{c}(A,\phi)$ approaches the line spanned by the eigenvector centrality $\textbf{v}(A)$ in the limit as $\phi$ approaches $\|A\|_2^{-1}$ from below.

To clarify the intuition behind results in later sections, we describe a stylized dynamic model relating eigenvector centrality to peer effects. Suppose that each individual begins with an endowment $w_i(0)$ in period $0$. In each period $t+1$, each individual's endowment changes to $c w_i(t) + \sum_{j \in N_i} w_i(t)$ for any $c>0$. This says that an agent's resources are a linear combination of her resources from the previous period and the sum of her neighbors' resources from the previous period. In matrix notation,
$$\textbf{w}(t+1) = (cI + A) \textbf{w}(t).$$
Then as $t \rightarrow \infty$, the endowments $\textbf{w}(t)$ become proportional to $\textbf{v}(A)$. The limit as endowments are repeatedly updated is the eigenvector centrality. This model closely resembles the mechanical process in Section V.A of \cite*{Echenique07}.

\section{Centrality on Random Networks}\label{thms}

This section characterizes the centralities of agents in networks generated using stochastic block models. With high probability, the vector of centralities will be close to a deterministic vector depending on the link probabilities. In Sections \ref{ineq} and \ref{gen}, we compute this deterministic vector for particular random network models. Combining those computations with the theorems in this section will clarify the impact of link probabilities, and thus features of network structure such as the amount of homophily, on the distributions of centralities of agents.

Before stating these theorems, we will state and discuss two technical conditions.

\begin{defin}\label{spectralgap}We say that a sequence of random networks has \textbf{non-vanishing spectral gap} if there exists $\delta > 0$ such that $$\overline{\lambda}_1(n) - |\overline{\lambda}_2(n)| > \delta \overline{\lambda}_1(n)$$
for all $n$.
\end{defin}

This says that the absolute value of the ratio between the second eigenvalue and the first is bounded away from one. \cite*{Golub12} call the second eigenvalue spectral homophily and show that this quantity measures the amount of homophily in a network.\footnote{\cite*{Golub12} consider matrices with first eigenvalue equal to one.} As a simple example of the connection, consider a stochastic block model in which all groups have the same size, the probability of a given link within groups is equal to $p_s$ and the probability of a link between groups is $p_d$. In this setting the ratio between the second eigenvalue and the first eigenvalue is equal to the usual Coleman homophily index (see Section II.C of \cite*{Golub12}).

So the spectral gap is non-vanishing when there is not extreme homophily. Intuitively, the condition is satisfied if for large $n$ the network cannot be split into two subgroups which are nearly disconnected. If all nodes have the same degree in $\bar{A}(n)$, we can make this claim more formal. Suppose the ratio between the second eigenvalue and the first eigenvalue converges to one. Then the uncoupling theorem from \cite*{hartfiel1998structure} implies that there exist subsets $M(n)\subset \{1,\ldots,n\}$ of agents such that
$$\frac{\sum_{i \in M(n), j \notin M(n)}\bar{A}(n)_{ij} }{\sum_{j} \bar{A}(n)_{ij} } \rightarrow 0$$
as $n \rightarrow \infty$.
In words, the total weight on links between $M(n)$ and its complement grows much more slowly than the degree of a single agent.

A non-vanishing spectral gap ensures that the eigenvector centralities of $\bar{A}$ and $A$ are uniquely defined with high probability. In cases of extreme homophily, there can be multiple eigenvectors of $\bar{A}$ or $A$ with the largest eigenvalue corresponding to different disconnected groups or collections of groups. With non-vanishing spectral gap, we avoid this situation with high probability.

To state the second condition, we define $\Delta = \max_i \sum_{j = 1}^n \bar{A}(n)_{ij} $ to be the maximum expected degree of a node.
\begin{defin}
We say that a sequence of random networks has \textbf{large enough eigenvalues} if $$\frac{\overline{\lambda}_1(n)}{\sqrt{\Delta \log(n)}} \rightarrow \infty$$ as $n\rightarrow \infty$.
\end{defin}

The intuitive content of this technical condition is that (1) the expected degrees of agents grow sufficiently faster than $\log n$, and (2) expected degrees do not vary too much across agents. To better understand this intuition, we use the following fact (\cite*{Lev14}):
\begin{equation}\label{leveq}\overline{\lambda}_1(n)^2 \geq (\min_{i,j}\bar{A}(n)_{ij}) \sum_{i} \sum_j \bar{A}(n)_{ij}.\end{equation}
If each agent $i$'s link probabilities do not vary too much across neighbors $j$, condition (1) implies that $n \min_{i,j}\bar{A}(n)_{ij}$ grows faster than $\log n$.\footnote{More generally, a sequence can have large eigenvalues even if some links have probability zero or very small probability. In these cases equation (\ref{leveq}) will not give a useful lower bound on the eigenvalues.} If condition (2) holds, then $\sum_{i} \sum_j \bar{A}(n)_{ij}$ is not too large relative to $\Delta n$. Combining these two informal bounds then gives that $\overline{\lambda}_1(n)^2 > \Delta\log(n ).$ We will see how a sequence of random networks can fail to have large enough eigenvalues if one agent has a much higher expected degree than others in Example~\ref{smalleig}.

The condition rules out very sparse networks. If most or all agents have very few links, centralities will be very sensitive to whether particular links form. So characterizing centralities with high probability will not be possible. In practice, the condition may fail for networks where links are reserved for strong connections such as close friendships.

In Sections~\ref{two} and~\ref{comparative}, we study sequences of stochastic block networks for which each group's fraction of the total population remains fixed as $n$ grows, and all link probabilities $p_{ij}$ also remain fixed. It is easy to see these sequences have large enough eigenvalues.

Under these conditions we can characterize eigenvector centrality:
\begin{thrm}\label{thm1}
Suppose $A(n)$ is a sequence of random networks that has non-vanishing spectral gap and has large enough eigenvalues.  Let $\epsilon > 0$. For $n$ sufficiently large, with probability at least $1-\epsilon$ the matrix $A(n)$ has a unique largest eigenvalue $\lambda_1(n)$ and the eigenvector centralities $\textbf{v}(A(n))$ satisfy
$$\|\textbf{v}(A(n)) - \textbf{v}(\bar{A}(n))\|_2 < \epsilon.$$
\end{thrm}

The theorem says that in large networks, the eigenvector centralities of all agents approach their expected values, and these expected values can be computed from the probabilities of links. Notably, the distance between an agent's eigenvector centrality and its expected value is bounded uniformly (across agents).

The result has the flavor of the law of large numbers, and that theorem gives some intuition for why the bounds hold. By the law of large numbers, for $n$ big enough most agents have approximately the expected number of links to each other group. This suggests their eigenvector centralities will also be close to their expected values, though an actual proof requires more substantive random graph theory.

The key tool is a result from \cite*{Chung11} which bounds the norm of the matrix $A(n) - \bar{A}(n)$ and the difference between the eigenvalues of $A(n)$ and $\bar{A}(n)$. We show that with large enough eigenvalues, their result implies that for $n$ large, the matrix norm $\|A(n) - \bar{A}(n)\|_2$ is small and the difference $|\lambda_1(n)-\overline{\lambda}_1(n)|$ is small compared to $\overline{\lambda}_1(n)$.

Given these bounds, we can show that $\bar{A}(n)$ and $A(n)$ (considered as linear operators on $\mathbb{R}^n$) map the eigenvector centrality $\textbf{v}(A(n))$ to nearby vectors. In particular, the image of $\textbf{v}(A(n))$ under $\bar{A}(n)$ has norm close to $\overline{\lambda}_1(n)$. Because we bound the spectral gap of $\bar{A}$ below, the image of a unit vector under $\bar{A}(n)$ only has norm close to the first eigenvalue when that vector is close to the first eigenvector $\textbf{v}(\bar{A}(n))$.

%We can see that the non-vanishing spectral gap condition is necessary in the case where the expected degree of each node is the same. If $\frac{\overline{\lambda}_2(n)}{\overline{\lambda}_1(n)}\rightarrow 1,$
%then the argument after Definition~\ref{spectralgap} shows that for $n$ large the network $\bar{A}(n)$ is the union of two subgroups with very few connections between them. Perturbing the link probabilities slightly to add links in one of these subgroups, we can find a nearby matrix of link probabilities (in the operator norm) for which the eigenvector centralities of agents in that group are much higher. We can do the same for the other group. By continuity, Theorem~\ref{thm1} cannot hold for all convex combinations of these two matrices of link probabilities.

%The non-vanishing spectral gap condition is necessary in the following sense: given any sequence of random networks $\bar{A}$ with large enough eigenvalues such that $\lambda_2/\lambda_1$

We next state the analogous result for Katz-Bonacich centrality.
\begin{thrm}\label{thm2}
Suppose $A(n)$ is a sequence of random networks that has large enough eigenvalues. Suppose $\phi(n)$ is a sequence of constants such that $$\limsup_n  \phi(n) \overline{\lambda}_1(n) < 1.$$ Let $\epsilon>0$. For $n$ sufficiently large, with probability at least $1-\epsilon$ the vectors of Katz-Bonacich centralities satisfy
$$\|\textbf{c}(A(n) ,\phi(n)) -\textbf{c}(\bar{A}(n) ,\phi(n))  \|_2 < \epsilon \sqrt{n}.$$
\end{thrm}

The statement is essentially the same as Theorem~\ref{thm1}, with eigenvector centrality now replaced by Katz-Bonacich centrality. So this result says that the Katz-Bonacich centralities of all in large networks are close to their expected values, and these expected values also depend only on $\bar{A}$. Katz-Bonacich centrality has norm $O(\sqrt{n})$ while eigenvector centrality is defined to have norm $1$, so there is an extra factor of $\sqrt{n}$ in Theorem~\ref{thm2}.

The proof uses the same bounds on the matrix norm $\|A(n) - \bar{A}(n)\|_2$. The idea is then to express the relevant Katz-Bonacich centralities as a sum of matrix powers, and to bound the differences between these matrix powers with the bounds on $\|A(n) - \bar{A}(n)\|_2$.

\textbf{Example and Simulations: }The simplest non-trivial example is the Erd\H{o}s-R\'{e}nyi graph with $\bar{A}(n)_{ij}=p$ for all $i,j$, and $n$. This sequence has non-vanishing spectral gap and large enough eigenvalues. For any $\epsilon>0$, Theorem~\ref{thm1} implies
$$\|\textbf{v}(A(n)) - \frac{1}{\sqrt{n}}\textbf{1}\|_2  <\epsilon$$
with high probability for $n$ large and Theorem~\ref{thm2} implies for suitable $\phi(n)$ that
$$\|\textbf{c}(A(n),\phi(n)) - \frac{1}{1-\phi(n)np}\textbf{1}\|_2  <\epsilon \sqrt{n}$$
with high probability for $n$ large.

We simulated $100$ such random networks with $p=\frac14$ and $\phi(n) = \frac{1}{2\overline{\lambda}_1(n)}.$ When $n=500$, the average value of $\|\textbf{v}(A(n)) -\textbf{v}(\bar{A}(n)) \|_2$ is $0.0773$  and the average value of $\|\textbf{c}(A(n) ,\phi(n)) -\textbf{c}(\bar{A}(n) ,\phi(n))  \|_2$ is $1.757.$ When $n=1000$, the average values are $0.0548$ and $1.744$, respectively.

\textbf{Extensions: }In Appendix~\ref{ext}, we provide several extensions of Theorems~\ref{thm1} and~\ref{thm2}. Subsection~\ref{weights} allows non-integer edge weights by replacing the Bernoulli random variables $A_{ij}$ with uniform random variables. Subsection~\ref{clustering} allows for clustering, an increased propensity for two agents $i$ and $j$ to be connected if both are connected to a common neighbor $k$, in this setting with edge weights. We give versions of the two theorems for networks placing additional weight on triangles along with dyads. Because many real-world social and economic networks exhibit more clustering than independent formation of links would imply, this modification helps accommodate many applications.

\textbf{Counterexamples: }We now give two counterexamples showing that a non-vanishing spectral gap and large enough eigenvalues are necessary for the two theorems. In the first example, the eigenvalues of $A(n)$ are well-behaved but the eigenvector centrality depends discontinuously on the realized network.

\begin{example}
Suppose that $\bar{A}(n)_{ij} = \frac12$ if $i,j \leq \frac{n}2$ or $i,j >\frac{n}2$ and $\bar{A}(n)_{ij}=n^{-3}$ otherwise. The network is divided into two groups, and two agents in the same group are connected with probability $\frac12$ while two agents in separate groups are connected with very small probability. A simple calculation shows that $\frac{\overline{\lambda}_2(n)}{\overline{\lambda}_1(n)}\rightarrow 1$, so that the spectral gap does vanish.

Then with probability converging to $\frac12$ as $n \rightarrow \infty$, the eigenvector centralities of agents $1,2,\ldots,\lfloor \frac{n}{2}\rfloor$ are all $0$ while the eigenvector centralities of the remaining agents are positive. With probability also converging to $\frac12$ as $n \rightarrow \infty$, the eigenvector centralities of agents $\lfloor \frac{n}{2}\rfloor+1,\lfloor \frac{n}{2}\rfloor+2,\ldots,n$ are all $0$ while the eigenvector centralities of the remaining agents are positive. On the other hand ${v}_i(\bar{A}(n))=\frac{1}{\sqrt{n}}$ for all $i$, so 
$$\|\textbf{v}(A(n)) - \textbf{v}(\bar{A}(n))\|_2$$
cannot vanish with high probability. With high probability there are no edges between the two groups, and when this occurs, the eigenvector centrality can correspond to either of the two corresponding components of the network. Note that the example does not depend on multiplicity of the largest eigenvalue of $\bar{A}(n)$ or $A(n)$, but only on the first two eigenvalues of these matrices being close.
\end{example}

We observed above that a sequence of random networks has large enough eigenvalues if (1) the degrees of agents grow sufficiently faster than $\log n$ and (2) no small group of agents has too many links. It is easy to construct counterexamples when all agents' degrees are bounded. The next example shows the theorems fail even when degrees grow quickly if links are too concentrated.

\begin{example}\label{smalleig}
Suppose that $\bar{A}(n)_{ij}=\frac12$ if $i=1$ or $j=1$ and $\bar{A}(n)_{ij} = \frac{\log n}{n}$ otherwise.\footnote{We could also take  $\bar{A}(n)_{ij} =n^{\alpha}$ for $\alpha <-\frac12$.} The sequence fails to have large enough eigenvalues: the largest eigenvalue $\overline{\lambda}_1(n)$ grows at rate at most $\sqrt{2n}$ while the maximum expected degree is $\Delta = (n-1)/2$, so $$\frac{\overline{\lambda}_1(n)}{\sqrt{\Delta \log(n)}} \rightarrow 0.$$

Let $\phi(n) = \frac{\phi_0}{\sqrt{2n}}$ for $\phi_0 < 1$. By the law of large numbers, with high probability for $n$ large, agent $1$ has degree approximately $n/2$ while all other agents have degree approximately $\log n$. Conditional on this event, the number of paths of length $k$ beginning at agent $i$ with a link to agent $1$ is approximately
$$(n/2)^j + o(n^j) \text{ for } k=2j \text{ and } (j+1)(n/2)^j\log n + o(n^j \log n)\text{ for } k=2j+1.$$
The number of such paths in the weighted network with adjacency matrix $\bar{A}(n)$ is
$$(n/4)^j + o(n^j)\text{ for } k=2j \text{ and } (j+1)(n/4)^j\log n + o(n^j \log n)\text{ for } k=2j+1.$$
The number of paths beginning at agent $i$ depends substantially on the realization of the potential link between $1$ and $i$.

A simple computation then shows that we can bound $|c_i(A(n),\phi(n))-c_i(\bar{A}(n),\phi(n))|$ below by a constant independent of $n$ for all $i$ such that $A_{ij}(n)=1$. Thus there exists $\epsilon > 0$ such that for $n$ large
$$\|\textbf{c}(A(n) ,\phi(n)) -\textbf{c}(\bar{A}(n) ,\phi(n))  \|_2 > \epsilon \sqrt{n}$$
with high probability, so Theorem~\ref{thm2} is violated.
\end{example}

\textbf{Rate of Convergence: }While the assumptions we make are not sufficient for our methods to give guarantees about the rate of convergence of eigenvector and centrality measures, our analysis can give results about the rate of convergence under stronger assumptions. More precisely, suppose we have an upper bound on the rate at which $$\frac{\overline{\lambda}_1(n)}{\sqrt{\Delta \log(n)}} \rightarrow \infty.$$
Then one can derive an upper bound on the rate of convergence of the centrality measures from the proofs of Theorems~\ref{thm1} and~\ref{thm2}.

We make this explicit for Theorem~\ref{thm1}. Under the conditions of the theorem, if we let $f(n) = \sqrt{\Delta \log n} /\overline{\lambda}_1(n),$ then it follows from our proof that
$$\|\textbf{v}(A(n)) - \textbf{v}(\bar{A}(n))\|_2 < \frac{8 (f(n)-f(n)^2)}{2 \delta - \delta^2} < \frac{8f(n)}{\delta}$$
with high probability. For example if the degrees of all agents are linear in $n$, then this gives an upper-bound of $O(\sqrt{\frac{\log n}{n}})$ for the distance between the eigenvector centralities.

\section{Homophily and Inequality}\label{ineq}

This section elucidates the connections between homophily in network formation and inequality of outcomes. To do so requires a network model where we can vary the amount of homophily, so we consider large stochastic block networks (defined in Section~\ref{Model}).

Our goal is thus to understand which agents are central in stochastic block models. The first step is applying the main theorems from the previous section,  which reduce computing these centralities to analyzing the matrix of link probabilities $\bar{A}$. To apply the theorems, we now carry out this analysis of $\bar{A}$ for several random network models. We can thus study the relationship between the parameters governing link formation and distributions of centrality. 

We first compare distributions according to Lorenz dominance, which measures relative inequality. This analysis explains when policy changes affecting network structure will lead to more equal distributions of outcomes. We then ask which groups benefit most in absolute terms from changes to the network, i.e. which groups would be most affected by policy changes.

\subsection{Relative Inequality}\label{two}

We consider a sequence of stochastic block models with $m$ groups and fix all relative group sizes, so that each group $i$ has size $s_in$ for some constants $s_i$ summing to one and independent of $n$.\footnote{If the sizes $s_in$ are not integers, we can round the group sizes by any convention giving total population size $n$.}

Suppose there are just two probabilities of link formation, so the probabilities of link formation are $p_s = p_{ii}$ for all $i$ and $p_d = p_{ij}$ for all $i \neq j$, where $p_s > p_d$.
%The analysis would be unchanged if the link probabilities were replaced by $p_s f(n)$ and $p_d f(n)$ for any function such that $f(n)/\log n \rightarrow \infty.$

With two groups, we will assume that $s_1>s_2$ so that group $1$ is the majority group. In the language of the education example from the introduction, the model in this case describes a school with black and white students. Suppose a majority share  $s_1$ of students are white and that any two students of the same race are friends with probability $p_s$ while friendships between races form with probability $p_d$. We can interpret the results in this section as describing which schools we should expect to have more equal educational achievement based on their social networks.

We establish dominance results using the following criterion.

\begin{defin}
Given two distributions $x= (x_1,...,x_n)$ with $x_1 \leq x_2 \leq ... \leq x_n$ and $ y= (y_1,...,y_n)$ with $y_1 \leq y_2 \leq ... \leq y_n$, we say that the distribution $x$ \textbf{Lorenz dominates} $y$ if $$\frac{\sum_{i = 1}^k x_i}{\sum_{i = 1}^n x_i} \geq \frac{\sum_{i = 1}^k y_i}{\sum_{i = 1}^n y_i}$$
for all $1 \leq k \leq n$.
\end{defin}

In words, distribution $x$ Lorenz dominates distribution $y$ if for each $k$, the share of total resources held by the poorest $k$ individuals in distribution $x$ is at least the share held by the poorest $k$ individuals in distribution $y$. Geometrically, the definition says that the Lorenz curve of distribution $x$ (which plots the share of resources held by the poorest $k$ individuals) lies above the Lorenz curve of distribution $y$.% This is illustrated in Figure~\ref{fig:Lorenz}.

%\begin{figure}
%\center{\includegraphics[scale=.75]{LorenzExample.png}}
%
%\caption{Two Lorenz curves. The blue distribution Lorenz dominates the red distribution.}\label{fig:Lorenz}
%\end{figure}

Lorenz domination gives a partial order on distributions which nests a wide family of measures of inequality. Most notably, if distribution $x$ Lorenz dominates distribution $y$ then $x$ also has a smaller Gini coefficient than $y$.

For the following proposition, let $\bar{A}(n)$ be the matrix of link formation probabilities with within-group probability $p_s$, between-group probability $p_d$ and group sizes $s_1,\hdots,s_n$. Similarly let $\bar{A}'(n)$ be the matrix of link formation probabilities with within-group probability $p_s'$, between-group probability $p_d'$ and majority group size $s'_1,\hdots,s'_n$. In part (i), we compare distributions of centralities as we vary link formation probabilities fixing group sizes ($s_i=s'_i$). In part (ii), we vary group sizes while keeping link formation probabilities fixed ($p_s=p_s',p_d=p_d'$):

\begin{proposition}\label{twoeig}

With probability approaching $1$ as $n \rightarrow \infty$:

(i) the eigenvector centralities $\textbf{v}(A(n))$ Lorenz dominate the eigenvector centralities $\textbf{v}(A'(n))$ if $A'(n)$ has a higher within-group link formation probability ($p_s' \geq p_s$) and  lower between-group link formation probability ($p_d' \leq p_d$) than $A(n)$.

(ii) with two groups, the eigenvector centralities $\textbf{v}(A(n))$ Lorenz dominate the eigenvector centralities $\textbf{v}(A'(n))$ if $A'(n)$ has a larger majority group ($s_1 \leq s'_1$) and both group sizes $s_1,s_1' \leq \overline{s}$ for a constant $\overline{s}  \in (\frac12, 1)$ depending on only $p_s$ and $p_d$.

\end{proposition}

Increasing $p_s$ and decreasing $p_d$ both correspond to increasing homophily. So part (i) says that more segregated networks lead to more unequal outcomes. Special cases include fixing one of the parameters $p_s$ or $p_d$, as well as varying the two parameters so that the total number of links stays fixed while the fraction of links within groups increases.

We use the formal connection between eigenvector centrality and the dynamic process of updating from Section \ref{EigCent} to give intuition for this result. Recall that for any non-zero $\textbf{w}(0)$, $A^t \textbf{w}(0)$ is approximately proportional to $\textbf{v}(A)$ for $t$ large. So we can consider repeatedly updating individuals' endowments from ${w}_i(t-1)$ to $c {w}_i(t) + \sum_{j \in N_i} {w}_i(t)$. With more links within groups, resources flow to larger groups. With more links between groups, the flows of resources are more equal.

Part (ii) says that when there are two groups and neither group is too large, a bigger majority group means more inequality. There are two effects here. The first is that as the minority group becomes smaller, its members do worse. But the second effect is that as the minority group grows smaller, there are more individuals in the better off group. When group sizes are in the range $(\frac12,\overline{s})$, the first effect is more important.

\begin{figure}
\center{\includegraphics[scale=.55]{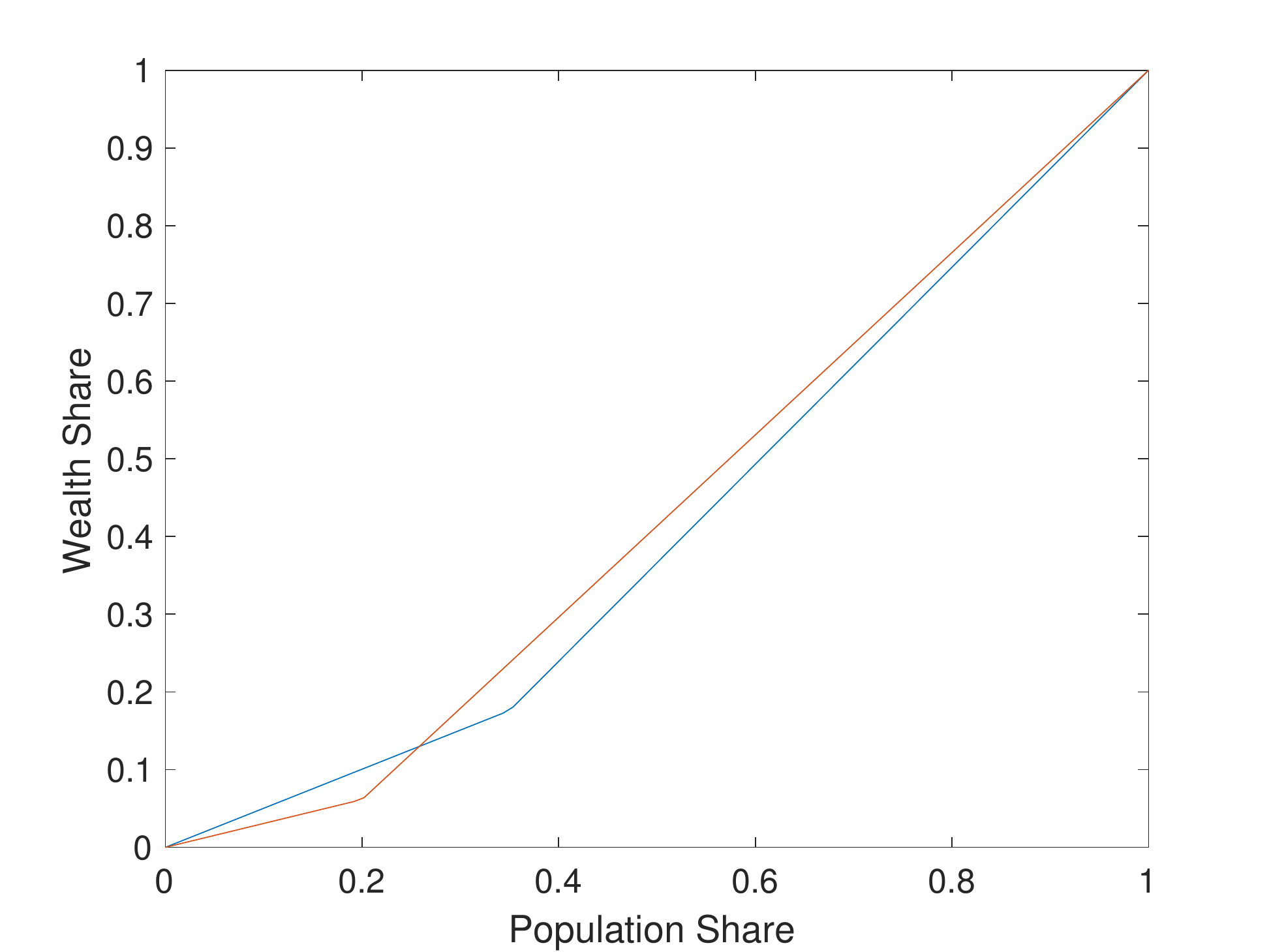}}

\caption{Lorenz curves for eigenvector centralities of two groups with $p_s = .5, p_d = .1$, $s_1 = .65$ in blue, and $s_1' = .8$ in red. Note that neither distribution Lorenz dominates the other.}\label{fig:groupsize}
\end{figure}

When a majority group is much larger than the corresponding minority group, though, we can obtain two distributions that cannot be compared using Lorenz dominance. The smaller minority group does worse, but is also smaller, so we cannot determine which distribution is more unequal (without using an ordering on distributions that is more complete than the Lorenz). Figure~\ref{fig:groupsize} shows an example, with link parameters $p_s=.5$ and $p_d=.1$ and majority group population shares $s= .65$ in blue and $s=.8$ in red.

Using Theorem~\ref{thm1}, the proof is reduced to comparing the eigenvector centralities of $\bar{A}(n)$ and $\bar{A}'(n)$. To do so, we use our recursive definition of eigenvector centrality to compute the elasticities $$\frac{\partial}{\partial p_s}\left( \frac{v_i(\bar{A}(n)) }{v_{i'}(\bar{A}(n)) }\right)=-\frac{\partial}{\partial p_d}\left( \frac{v_i(\bar{A}(n)) }{v_{i'}(\bar{A}(n)) }\right)=\frac{\overline{\lambda}_1^{-1}(s_{l}-s_{l'})n}{(1-\overline{\lambda}_1^{-1}s_{l}n (p_s-p_d))^2},$$
Here $s_l$ and $s_l'$ are the population shares of the groups containing $i$ and $i'$, respectively. This formula shows that the ratio between the centralities of a member of a larger group and a member of a smaller group is increasing in $p_s$ and decreasing in $p_d$. We show this implies Lorenz dominance as we decrease $p_s$ or increase $p_d$.

Finally, part (i) of the proposition applies with eigenvector centrality replaced by Katz-Bonacich centrality. We now fix group sizes.

\begin{proposition}\label{twobon}

Choose a sequence of $\phi(n)$ such that $$\limsup_n  \phi(n) \overline{\lambda}_1(n) < 1 \text{ and }\limsup_n  \phi(n) \overline{\lambda}'_1(n) < 1 .$$ With probability approaching $1$ as $n \rightarrow \infty$, the Katz-Bonacich centralities $\textbf{c}(A(n),\phi)$ Lorenz dominate the Katz-Bonacich centralities $\textbf{c}(A'(n),\phi)$ if $A'(n)$ has a higher within-group link formation probability ($p_s' \geq p_s$) and  lower between-group link formation probability ($p_d' \leq p_d$) than $A(n)$.
\end{proposition}

The proof and intuition are similar to Proposition~\ref{twoeig}. The relevant elasticities are now
$$\frac{\partial }{\partial p_s}\left( \frac{c_i(\bar{A},\phi)}{c_j(\bar{A},\phi)} \right)= -\frac{\partial }{\partial p_d} \left( \frac{c_i(\bar{A},\phi)}{c_j(\bar{A},\phi)}
\right) =\frac{(s_l-s_{l'})(p_s-p_d)\phi n}{(1-s_l(p_s-p_d)\phi n)^2},$$
where $s_l$ and $s_l'$ are the population shares of the groups containing $i$ and $i'$. The basic intuition about group sizes is also similar, but there are not clean results because the analog of the cutoff $\overline{s}$ depends on $n$ and on the choice of constant $\phi(n)$.

\subsection{Comparisons between Comparative Statics}\label{comparative}

We now focus on comparative statics of Katz-Bonacich centrality and ask which groups benefit more in absolute terms from changes in link parameters.\footnote{Related results have been derived for comparative statics of the first eigenvector, e.g. \cite{magnus1985differentiating} and \cite{conlisk1985comparative}. Stronger conditions are needed to determine the sign of the derivatives of eigenvector centrality due to the normalization.} Even with only a within-group and a between-group link probability, changes in link parameters can have more complicated effects on levels of centrality. We find that adding within-group links still benefits larger groups, but adding between-group links can have ambiguous effects. We then give an example showing that without the restriction to two link probabilities, extra links between group $1$ and group $2$ can benefit other well-connected groups more than group $1$.

In Appendix~\ref{appendixstatics}, we provide two formulas for the comparative statics of Katz-Bonacich centrality as some link probability $p_{ij}$ changes. The first expresses the derivative of $c_i(A,\phi)$ as a suitably weighted count of the number of walks beginning at agent $i$. The second is a more explicit expression depending on the link probabilities $p_{ij}$, and is easy to compute in examples. This result replace each group with a representative agent and then count walks in the network of representative agents. These formulas are used to prove the propositions in this section, and may also be useful tools in other settings.

We again consider stochastic block models with $m$ groups and relative group sizes fixed at $s_in$.\footnote{We again round the group sizes by any convention giving the appropriate total population.} The next two propositions maintain the assumption that there are two link probabilities $p_s$ and $p_d$, depending on whether the relevant nodes are in the same group or different groups. 

To facilitate simple statements of results, all comparative statics are first given for deterministic networks $\bar{A}$ with non-integral weights. The corresponding results for sequences of random networks follow as corollaries. 

Let $\frac{\partial}{\partial p_s} \textbf{c}(\bar{A},\phi)$ denote the derivative of Katz-Bonacich centrality as the entries of $\bar{A}$ corresponding to within-group links vary.

\begin{proposition}\label{same}
Suppose that $\phi < \overline{\lambda}^{-1}$. Then $$\frac{\partial}{\partial p_s} c_i(\bar{A},\phi) > \frac{\partial}{\partial p_s} c_j(\bar{A},\phi).$$
whenever $i$'s group is larger than $j$'s group.
\end{proposition}

When more links are added within groups, agents in larger groups benefit more than agents in smaller groups. Because $c_i(\bar{A},\phi)$ is increasing in $i$'s group size in this setting, this means that adding connections within groups increases inequality (in absolute terms).

There are two effects contributing to this result. First, increasing $p_s$ increases the degree of agents in larger groups more than the degree of agents in smaller groups. Second, because agents in larger groups are more central already, these agents are better positioned to benefit from extra links in other areas of the network. By contrast, when $p_d$ changes we will see these two effects work in opposite directions.

We next consider a sequence of random networks. Let $\bar{A}(n)$ be the matrix of link probabilities arising from these group sizes, within-group link probability $p_s$ and between-group link probability $p_d$. Define $\bar{A}'(n)$ similarly, with probabilities $p_s'$ and $p_d'$.

\begin{corollary}\label{pscor}
Suppose that $p_s'> p_s$ and $p_d' =p_d$. Fix a sequence of $\phi(n)$ bounded away from $\overline{\lambda}_1(n)$ and $\overline{\lambda}_1'(n)$. Then with probability approaching $1$ as $n \rightarrow \infty$,
$$c_i({A'(n)},\phi(n)) - c_i({A}(n),\phi(n)) > c_j({A}'(n), \phi(n))-c_j({A}(n),\phi(n))$$
whenever $i$'s group is larger than $j$'s group.
\end{corollary}

Using Theorem~\ref{thm2}, we can reduce from random networks to deterministic weighted networks. Then the result is a straightforward application of Proposition~\ref{same}.

When we vary the probability $p_d$ of links between groups, results are more ambiguous. Let $\frac{\partial}{\partial p_d} \textbf{c}(\bar{A},\phi)$ denote the derivative of Katz-Bonacich centrality as the entries of $\bar{A}$ corresponding to between-group links vary. 

\begin{proposition}\label{dif}
There exist constants $0 < \underline{\phi} < \overline{\phi}<\overline{\lambda}^{-1}$ such that:

(i) for $0 < \phi < \underline{\phi}$,
$$\frac{\partial}{\partial p_d} c_i(\bar{A},\phi) > \frac{\partial}{\partial p_d} c_j(\bar{A},\phi).$$
whenever $i$'s group is smaller than $j$'s group.

(ii) for $\overline{\phi}< \phi < \overline{\lambda}^{-1}$,
$$\frac{\partial}{\partial p_d} c_i(\bar{A},\phi) > \frac{\partial}{\partial p_d} c_j(\bar{A},\phi).$$
whenever $i$'s group is larger than $j$'s group.
\end{proposition}

There are now two opposing effects, so extra links between groups can help smaller groups more or larger groups more. First, increasing $p_d$ now increases the degree of those in smaller groups more than the degree of those in larger groups. But second, agents in larger groups are still better positioned to take advantage of extra links at a distance.

When $\phi$ is small, Katz-Bonacich centrality primarily counts shorter walks. So the first effect wins out, and smaller effects benefit more. When $\phi$ is close enough to $\overline{\lambda}_1^{-1}$, Katz-Bonacich centrality is dominated by walks of very long length. Then the first effect is tiny while the second effect is large. So agents in a central majority group actually benefit more from extra connections between groups than those in a disadvantaged minority.

The first result is a straightforward consequence of the relationship between degree centrality and Katz-Bonacich centrality with $\phi$ small. When $\phi$ is small, Katz-Bonacich centrality is approximately equal to an affine transformation of degree centrality, and the comparative static is easy for degree centrality.

For the second result, the Katz-Bonacich centralities are equal to suitably weighted counts of the number of walks beginning in each group. We can compute the derivatives of Katz-Bonacich centrality by adding a very small number of extra links and counting the walks which pass through these new links (see Appendix~\ref{appendixstatics}). When $\phi$ is large, these counts are dominated by very long walks, and we claim that for very long walks the starting group has little effect on which groups the walk passes through. To show the claim formally, we represent a random walk as a Markov chain with a state corresponding to each agent and transition probabilities proportional to link probabilities, and then study the stationary distribution of this Markov chain. Once the claim is shown, we can compute the comparative statics of interest by counting the total number of walks starting from a given agent, and this number is higher for members of larger groups.

For the following Corollary, we define $\bar{A}(n)$ and $\bar{A'}(n)$ as in Corollary~\ref{pscor}.

\begin{corollary}\label{pdcor}
Suppose that $p_s = p_s'$ and $p_d'> p_d$. Fix a sequence of $\phi(n)$ bounded away from $\overline{\lambda}_1(n)$ and $\overline{\lambda}_1'(n)$. Then there exist constants $\underline{C},\overline{C} \in (0,1)$ such that with probability approaching $1$ as $n \rightarrow \infty$,

(i) if $\phi(n) < \underline{C}  \cdot \overline{\lambda}_1(n)^{-1} $ for all $n$, then $$c_i({A'}(n),\phi(n)) - c_i({A}(n),\phi(n)) > c_j({A}'(n), \phi(n))-c_j({A}(n),\phi(n))$$
whenever $i$'s group is smaller than $j$'s group.

(ii) if $\phi(n) > \overline{C} \cdot \overline{\lambda}_1(n)^{-1}$ for all $n$, 
$$c_i({A'}(n),\phi(n)) - c_i({A}(n),\phi(n)) > c_j({A}'(n), \phi(n))-c_j({A}(n),\phi(n))$$
whenever $i$'s group is larger than $j$'s group.
\end{corollary}

The preceding results rely on the assumption that there are only two link probabilities $p_s$ and $p_d$, depending on whether a pair of agents are in the same group or in separate groups. With more link probabilities, a wide range of comparative statics are possible. We now give one interesting example.

Let $\frac{\partial }{\partial p_{ij}} \textbf{c}(\bar{A},\phi)$ be the derivative of Katz-Bonacich centrality as the entries of $\bar{A}$ corresponding to links between group $i$ and group $j$ vary.

\begin{proposition}\label{exam}There exist a constant  $\underline{\phi} < \overline{\lambda}_1$, link-formation probabilities $\bar{A}$ and groups $i,$ $j$ and $k$ such that
$$\frac{\partial}{\partial p_{ij}} c_k(\bar{A},\phi) > \frac{\partial}{\partial p_{ij}} c_i(\bar{A},\phi)$$
for all $\underline{\phi} < \phi < \overline{\lambda}_1$.
\end{proposition}
Adding links between groups $i$ and $j$ increases the Katz-Bonacich centrality of agents in group $k$ more than those in group $i$. This is an indication of the power of central network position: group $k$ is so well connected that its members are better positioned to take advantage of new links than the agents in group $i$ actually forming those links.

To prove the proposition, we give an example with three groups such that connections within group $1$ are very dense, connections between groups $1$ and $2$ are somewhat sparse, and all other connections are very sparse (see Figure~\ref{fig:threegroup}).

\begin{figure}
\center{\includegraphics[scale=.8,trim=15cm 17cm 15cm 5cm]{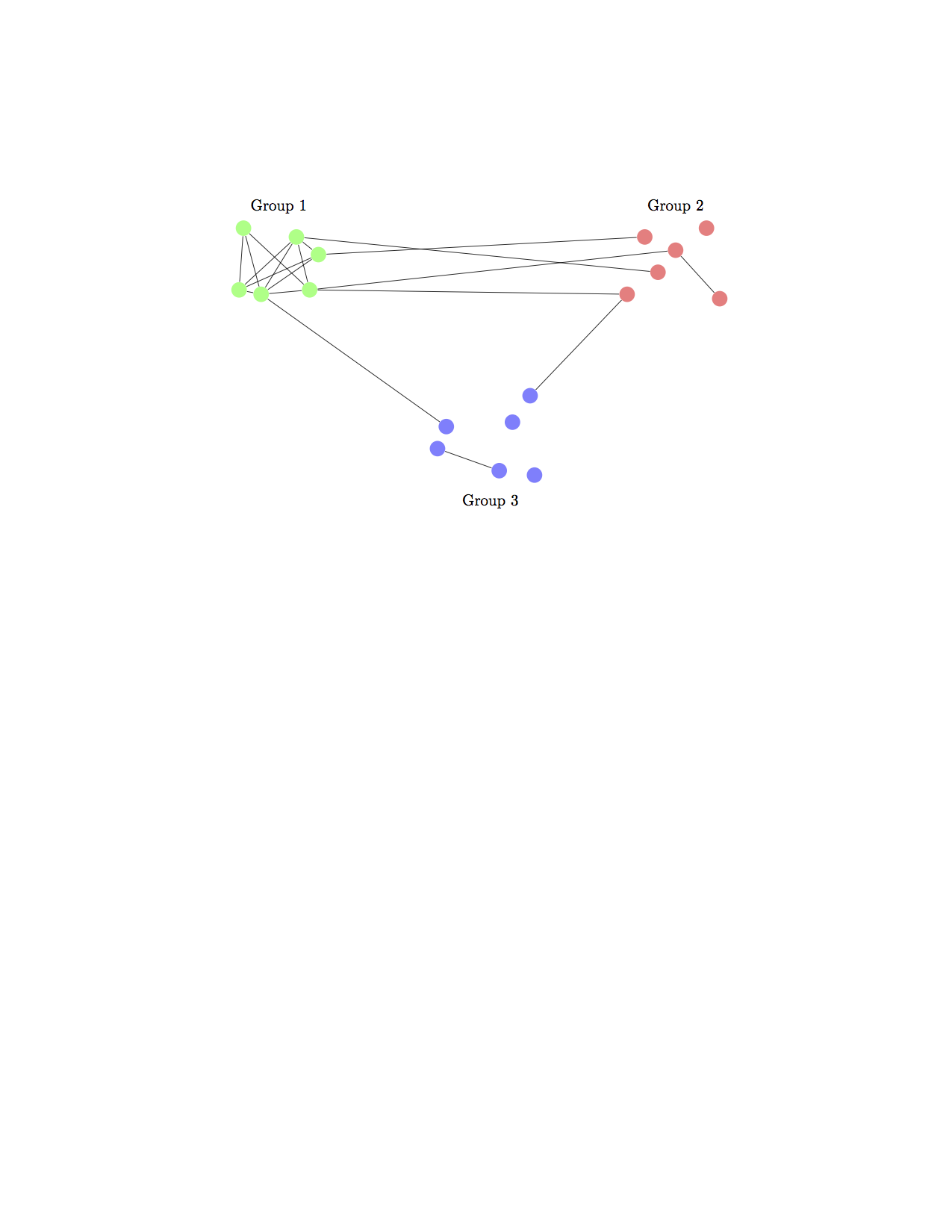}}

\caption{Illustration of a random network with link probabilities as in the proof of Proposition~\ref{exam}.}\label{fig:threegroup}
\end{figure}

Then if extra links are added between groups $2$ and $3$, group $2$ adds more short connections than group $1$. But group $1$ adds more longer walks. This is because there are many long walks that begin in group $1$, remain there for many links, and then finish by crossing to group $2$ and then group $3$. But most long walks beginning in group $2$ and spending most of their time in group $1$ must pass between groups $1$ and $2$ twice, which is less likely. So when $\phi$ is large, the Katz-Bonacich centrality of a member of group $1$ increases more than the Katz-Bonacich centrality of members of group $2$.

The example extends to actual random networks. We again consider $m$ groups with relative group sizes fixed, but now allow the entries $p_{ij}$ of $\bar{A}$ to be any constant probabilities. The entries of $\bar{A}'$ are equal to the entries of $\bar{A}$, except for a single pair of groups $i$ and $j$.

\begin{corollary}\label{excor}
There exists $p_{ij}' > p_{ij}$, a sequence of constants $\underline{\phi}(n)$ bounded away from $\overline{\lambda}_1(n)$ and $\overline{\lambda}_1'(n)$ and groups $i$, $j$ and $k$ such that with probability approaching $1$ as $n \rightarrow \infty$,
$$c_k({A'}(n),\phi(n)) - c_k({A}(n),\phi(n)) > c_i({A}'(n), \phi(n))-c_i({A}(n),\phi(n))$$
for any $\underline{\phi}(n) < \phi(n) < \overline{\lambda}_1(n)$.
\end{corollary}

\section{Other Network Formation Models}\label{gen}

The basic technique of the previous section was to compute centralities on random networks by applying the main theorems and then carrying out a deterministic computation. This method is more general, and in this section we consider applications to more general random network models beyond stochastic block models.

The first subsection discusses spatial models of network formation and gives a numerical example where certain agents are well-positioned when links at a substantial distance are very unlikely or very likely, but not for intermediate probabilities. The second subsection gives a model of random networks in which link formation probabilities depend on several characteristics. Under independence assumptions, the centralities can be computed ``characteristic by characteristic".

\subsection{Spatial Models}

A common alternative to network formation models based upon group structure relies instead on spatial structure (for a few examples, see \cite*{leung2015two}, \cite*{chaney2014network} and \cite*{breza2017using}). Agents are distributed in a continuous space with a distance metric, and closer agents are more likely to be connected. This space could represent physical locations in a geographic space or more generally some space of characteristics.\footnote{Deriving link probabilities from distance imposes a restriction due to the triangle inequality: if $i$ and $j$ are likely to be linked and $j$ and $k$ are likely to be linked, then the probability of $i$ and $k$ linking is not too low. This need not hold in stochastic block models.} Our tools can help determine how network centralities in these models depend on locations and underlying parameters.

In a latent space model based on \cite*{Hoff02}, the link between each pair of agents $i$ and $j$ forms with probability $\bar{A}_{ij}$ proportional to
$$\exp( \beta \cdot x_{i,j} + \gamma d(i,j)),$$
where $x_{i,j}$ are covariates of the pair, $\beta$ is a constant vector, $\gamma$ is a constant coefficient and $d(i,j)$ is the distance between $i$ and $j$. When all links form independently, our theorems apply to sequences of such networks. Our theoretical results thus complement \cite*{breza2017using}, who give evidence based on simulations and empirical data that centralities (as well as other network statistics) can be approximated well based on only the parameters of latent space models. Our Theorems~\ref{thm1} and~\ref{thm2} imply that these centralities can indeed be identified asymptotically given the underlying parameters when agents are spatially distributed such that the conditions of the theorems hold.\footnote{For example, if distances $d(i,j)$ are uniformly bounded, there is a uniform positive lower bound on $\bar{A}(n)_{ij}$ so both conditions hold.}

We give a numerical example in a similar model showing that certain agents are central when links decay at distance very quickly or very slowly, but not when the decay rate is intermediate. A functional form commonly used in empirical work sets the probability $\bar{A}_{ij}$ of a link between agents $i$ and $j$ to $d(i,j)^{-\rho}$. The parameter $\rho$ determines how quickly connections decay at a distance.

To define $\bar{A}$, consider agents arranged on an $(k+1) \times (k+1)$ grid in the plane, so that there is one agent at each point in $$\{(x,y): x,y \in \mathbb{Z}, 0 \leq x, y \leq k\}.$$ Each link between agents at distinct coordinates $(x,y)$ and $(x',y')$ has weight $d((x,y),(x',y'))^{-\rho}$ where $d$ is Euclidean distance, and each self-link has weight $1$. Figure~\ref{fig:Grid} illustrates an example with $k=20$ and $\rho=\frac12$.

\begin{figure}
\center{\includegraphics[scale=.6]{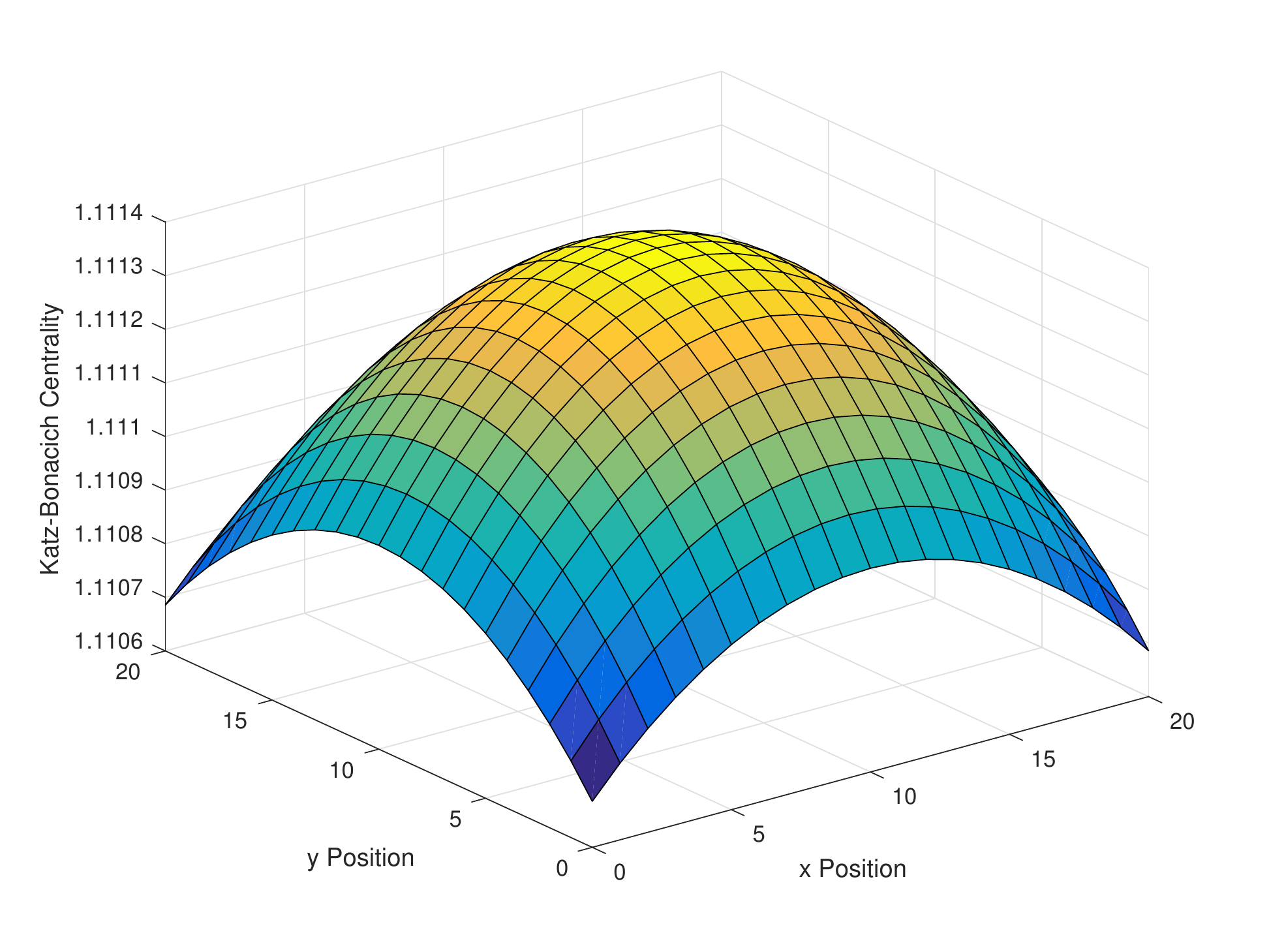}}

\caption{The Katz-Bonacich centralities of $\bar{A}$ for a grid with $k=20$, $\rho=\frac12$ and ${\phi}  \cdot \overline{\lambda}_1= \frac{1}{2 }$.}\label{fig:Grid}
\end{figure}

We consider eigenvector centralities in this network. The following claim is verified in the appendix:
\begin{claim}\label{gridclaim}
For $\rho \leq 1$, the sequence of random networks $\bar{A}(n)$ has large enough eigenvalues.
\end{claim}
\noindent Thus, Theorem~\ref{thm2} implies that the centralities on the corresponding random network are close for $k$ large.

As $\rho$ varies, the rankings between Katz-Bonacich centralities of various agents vary as well. To give an example, we take $k=20$ and set ${\phi} = \frac{1}{10}\cdot \overline{\lambda}_1^{-1}$. When $\rho=\frac{1}{2}$, the agent in position $(0,10)$ is more central than the agent in position $(3,3)$. But for $\rho$ sufficiently close to $1$, the agent in position $(3,3)$ is more central than $(0,10)$. The first agent is on the boundary in one direction, while the second is near the boundary in both directions and further from the central agent at $(10,10)$. When $\rho$ is large, potential neighbors nearby in the grid are most important and $(3,3)$ has many of these. As $\rho$ grows small, having many connections at intermediate distances becomes more valuable. Note that in the limit $\rho \rightarrow 0$ all agents become equally central, so the difference (or ratio) between the centralities is non-monotonic in $\rho$.

When $\rho$ is very small or large, the agent $(3,3)$ who is on the interior of a diagonal of the grid is relatively central. But when $\rho$ is intermediate, that agent suffers due to a small number of intermediate-length connections and the agent $(0,10)$ on the center of an edge of the grid is more central.
%We give a numerical example in a similar model. Consider the Katz-Bonacich centralities when agents are evenly distributed on the unit interval and have stronger connections to closer peers. There are $n$ agents with agent $i$ at location $\frac{i-1}{n-1}$, and the connection between distinct agents $i$ and $i'$ has weight $|i-i'|^{-\rho}$. The parameter $\rho$ determines how quickly connections decay at a distance. Explicitly,
%$$\bar{A}_{ij}=\min(|i-i'|^{-\rho},1).$$
%
%We find that the effect of varying $\rho$ on the ratios between centralities can be non-monotonic. When links do not decay at distance or decay very quickly, agents closer to an endpoint of the interval are almost as central as the middle agent. But when links decay at an intermediate rate, the centrality of those near the endpoint drops compared to the middle agent. Table~\ref{numex} gives the ratios of centralities between the fifth and sixth of eleven agents as well as the first and sixth of eleven agents, and we observe these ratios are non-monotonic in $\rho$.
%
%\begin{table}
%\centering
%\caption{Katz-Bonacich Centralities on the Unit Interval with $\phi=\frac{1}{10}$ and $n=11$ }
%
%\label{numex}
%\begin{tabular}{ | c | c | c | c | c | }
%\hline
%$\rho$ & 0 & $\frac12$ & 1 & 2
%\\ 
%\hline
%$c_{5}(\bar{A},\phi)/c_{6}(\bar{A},\phi)$ & 1 &    0.8483 & 0.8199 &  0.8430
%\\ 
%\hline
%$c_1(\bar{A},\phi)/c_{6}(\bar{A},\phi)$ & 1 &  0.9953 &    0.9952 & 0.9978
%\\ \hline
%\end{tabular}
%
%
%\end{table}
\subsection{Multicharacteristic Networks}

We can build a larger class of network formation models, which we call multicharacteristic random networks, by combining several characteristics. For example, we could model agents of different races or genders living in different neighborhoods by combining a stochastic block model with a model based upon geography. When characteristics are independent, centralities depend on the two characteristics in a simple, separable manner.

Let $\bar{A}^{(1)}$ and $\bar{A}^{(2)}$ be two matrices of link probabilities for networks of sizes $n_1$ and $n_2$, respectively. We will let $\otimes$ denote the Kronecker product. We define multicharacteristic random networks, which combine the two networks by multiplying link probabilities.

\begin{defin}
The \textbf{multicharacteristic random network} $\bar{A} = \bar{A}^{(1)}\otimes \bar{A}^{(2)}$ with layers $\bar{A}^{(1)}$ and $\bar{A}^{(2)}$ has $n_1n_2$ agents indexed by $(i_1,i_2)$ and link probabilities between $(i_1,i_2)$ and $(j_1,j_2)$ given by the product $\bar{A}^{(1)}_{i_1j_1}\bar{A}^{(2)}_{i_2j_2}$ of the corresponding link probabilities in the two layers.
\end{defin}

If $\bar{A}^{(1)}$ and $\bar{A}^{(2)}$ correspond to two characteristics in the network, then $\bar{A}$ is the network formed by assuming the two characteristics are independently distributed and links form when they would form in both layers independently. These assumptions are strong, but let us obtain clean expressions for centralities in the multicharacteristic network.

\begin{lemma}\label{layers}
The eigenvector centrality of agent $(i_1,i_2)$ in the multicharacteristic random network $\bar{A}$ is given by
$$v_{(i_1,i_2)}(\bar{A})= v_{i_1}(\bar{A}^{(1)})v_{i_2}(\bar{A}^{(2)}).$$
\end{lemma}

The lemma says that eigenvector centralities can be computed separately in each layer and multiplied together. Note that we can rephrase the result as stating that
$$\textbf{v}(\bar{A}) = \textbf{v}(\bar{A}^{(1)}) \otimes \textbf{v}(\bar{A}^{(2)}).$$ This standard fact can be used to show that centralities behave similarly on actual random networks. 

\begin{proposition}\label{layerscor}
Suppose $\bar{A}^{(1)}(n_1)$ and $\bar{A}^{(2)}(n_2)$ are sequences of random networks which both have non-vanishing spectral gaps and large enough eigenvalues. Let $\bar{A}(n_1n_2)$ be the corresponding sequence of multicharacteristic random networks. For any $\epsilon > 0$, when $n$ is sufficiently large
$$\|\textbf{v}(A(n_1n_2))-\textbf{v}(A^{(1)}(n_1))\otimes  \textbf{v}(A^{(2)}(n_2)\|_2 < \epsilon$$
with probability at least $1-\epsilon.$
\end{proposition}

So $v_{(i_1,i_2)}({A})$ is very likely to be close to $v_{i_1}({A}^{(1)})v_{i_2}({A}^{(2)})$ for each agent $(i_1,i_2)$. When link probabilities are multiplicative, so are centralities. The proof applies Theorem~\ref{thm1} and Lemma~\ref{layers}.

The results extend immediately to networks with more than two characteristics.

\section{Conclusion}\label{conc}

We gave asymptotic characterizations of the eigenvector and Katz-Bonacich centralities of agents in random networks. Theoretically, these results simplify questions about how changing network structure affects who is central. We provided several applications: (1) In stochastic block models, we showed the interaction between network segregation and differences in group sizes increases inequality. (2) We could calculate and compare the network centralities of agents in various locations when links depend on spatial locations. (3) When networks depend on several independent characteristics, centralities could be computed by treating each characteristic separately.

We conclude by proposing several consequences of our techniques beyond these applications:
\begin{itemize}
\item Centrality can be studied in applied settings where full network data is unavailable: in large networks we can approximate agents' centralities well using only link formation probabilities for various groups.
\item Conversely, when an econometrician has access to network data but parameters in the underlying network formation process are unobserved, our theorems justify a structural model of network formation. The unobserved parameters can be estimated using centralities or related quantities as moments.
\item A large literature in industrial organization models pricing decisions as a quadratic game, with each firm competing with their neighbors in a network (\cite*{vives2010information}, \cite*{vives2017endogenous}). While this literature has focused on simple network structures, our methods could be used to solve for expected prices on random networks and study how network structure influences prices.
\end{itemize}
Like the applications developed in the present work, each relies on the basic method of reducing questions about random graphs to deterministic calculations.
\appendix

\section{Proofs}
We will sometimes refer to the Euclidean norm and the induced matrix norm as $\|\cdot\|$, omitting the subscripts $2$.

\begin{proof}[Proof of Theorem~\ref{thm1}]

The outline of the proof is as follows: the first step is to show that the matrix norm $\|A(n)-\bar{A}(n)\|_2$ is small with high probability. To show that this implies the corresponding eigenvector centralities are close, we observe that $\textbf{v}(A(n))$ is the vector in the unit sphere with the largest image under the linear operator $A(n)$. When the spectral gap is not too small, any other vector in the unit sphere with image almost as large under $A(n)$ must be close to $\textbf{v}(A(n))$. When $\|A(n)-\bar{A}(n)\|_2$ is small we can show that the eigenvector centrality of $\bar{A}(n)$ has large image under $A(n)$, and so must be close to the eigenvector centrality of $A(n)$.\footnote{An alternate approach to the eigenvector perturbation theory argument, which we produce directly here, is to appeal to the Davis-Kahan theorem (as in \cite*{avella2017centrality}.)}

Let $\epsilon > 0$. Recall we defined $\Delta = \max_i \sum_{j = 1}^n \bar{A}(n)_{ij} $ to be the maximum expected degree of a node. We sometimes drop arguments $n$ in the remainder of the proof, so for example $A(n)$ will be referred to as $A$.

It follows from the proof of Theorem 1 of \cite*{Chung11} that with probability at least $1-\epsilon$
\begin{equation}\label{eq:evbound}|\lambda_1 - \overline{\lambda}_1| \leq \sqrt{4 \Delta \log(2n/\epsilon)}.\end{equation}
In fact it follows from the proof of the theorem cited that with probability at least $1-\epsilon$
\begin{equation}\label{eq:normbound}\|A - \bar{A}\| \leq  \sqrt{4 \Delta \log(2n/\epsilon)}.\end{equation}
Note that the theorem cited uses our assumption that $A$ is symmetric. We will show that when these inequalities hold, $$\|\textbf{v}(A(n)) - \textbf{v}(\bar{A}(n))\|_2 < \epsilon.$$ 

Because the sequence has large enough eigenvalues, \begin{equation}\label{eq:largeev}\sqrt{4 \Delta \log(2n/\epsilon)}  \leq f(n) \overline{\lambda}_1\end{equation} for some sequence $f(n) \rightarrow 0$. So we have
\begin{align*} \|\bar{A} \textbf{v}(A)\| & \geq  \|A \textbf{v}(A)\| - \|A - \bar{A}\| \|\textbf{v}(A)\| & \text{by the triangle inequality}
\\ &  \geq \lambda_1  - f(n) \overline{\lambda}_1 & \text{by equations \ref{eq:normbound} and \ref{eq:largeev}}
\\ & = \overline{\lambda}_1 - (\overline{\lambda}_1-\lambda_1) - f(n) \overline{\lambda}_1 &
\\ & \geq \overline{\lambda}_1(1 - 2f(n)) & \text{ by equation \ref{eq:evbound}}.\end{align*}

On the other hand, we can write $\textbf{v}(A) = \alpha_1 \textbf{v}(\bar{A}) + \alpha_2 \textbf{w},$ where $\textbf{w}$ is a unit vector orthogonal to $\textbf{v}(\bar{A})$ and $\alpha_1^2 + \alpha_2^2 =1$ (where the coefficients $\alpha_1$ and $\alpha_2$ vary with $n$). Because the orthogonal complement to $\textbf{v}(\bar{A})$ is the span of the other eigenvectors of $\bar{A}$, we have $$\|\bar{A} \textbf{w} \| \leq \overline{\lambda}_2.$$
Because the sequence has non-vanishing spectral gap, there exists $\delta>0$ such that $|\overline{\lambda}_2| < (1-\delta)\overline{\lambda}_1.$ So
\begin{align*}\|\bar{A} \textbf{v}(A)\| & \leq \sqrt{(\alpha_1 \overline{\lambda}_1 )^2+ (\alpha_2 \overline{\lambda}_2)^2} \\ & \leq \overline{\lambda}_1 \sqrt{\alpha_1^2 +(1-\delta)^2\alpha_2^2}.  \end{align*}

Combining the upper and lower bound on $\|\bar{A}\textbf{v}(A)\|,$ we conclude that
$$\sqrt{\alpha_1^2 +(1-\delta)^2\alpha_2^2} \geq (1-2f(n)).$$
Because $f(n) \rightarrow 0$ and $\alpha_1^2 + \alpha_2^2 = 1$, this implies that $\alpha_2 \rightarrow 0$ as $n \rightarrow 0$. So taking $n$ sufficiently large, we conclude that $$\|\textbf{v}(A(n)) - \textbf{v}(\bar{A}(n))\| < \epsilon$$
with probability at least $1-\epsilon$.\end{proof}
\begin{proof}[Proof of Theorem~\ref{thm2}]
The outline of the proof is as follows: we show as in the proof of Theorem~\ref{thm1} that the matrix norm $\|A(n)-\bar{A}(n)\|_2$ is small with high probability. We then expand $\textbf{c}(A(n),\phi(n))$ as a sum of powers of the adjacency matrix. Using the bound on $\|A(n)-\bar{A}(n)\|_2$, we show these powers are also close in the matrix norm.

We have
$$\textbf{c}( \bar{A}(n),\phi(n)) = \sum_{k = 0}^{\infty} \phi(n)^k \bar{A}(n)^k \mathbf{1}_n,$$
and the analogous formula for $A(n)$.

Because $\limsup_n \phi(n)\overline{\lambda}_1(n) <1$, there  exists $K$ and $N$ such that $$\left\|\sum_{k=K}^{\infty} \phi(n)^k \bar{A}(n)^k\right\| \leq \sum_{k=K}^{\infty} (\phi(n) \overline{\lambda}_1(n))^k < \epsilon$$
whenever $n \geq N$. Increasing $N$ and decreasing $\epsilon$ if necessary, we can assume (by equations~\ref{eq:normbound} and~\ref{eq:largeev} from the proof of Theorem~\ref{thm1}) that
$$\|\phi(n) (\bar{A}(n)-A(n))\| < \epsilon/(K^2), $$ $$\left\|\sum_{k=K}^{\infty} \phi(n)^k \bar{A}(n)^k\right\| < \epsilon \text{ and }\left\|\sum_{k=K}^{\infty} \phi(n)^k {A}(n)^k\right\| < \epsilon$$
with probability at least $1-\epsilon$ whenever $n \geq N$. We condition on this event.

We have taken $N$ sufficiently large so that $\|\phi(n) \bar{A}(n)\|< 1$, and (taking $\epsilon$ smaller if necessary) it follows from the bound on $\|\phi(n) (\bar{A}(n)-A(n))\| $ that $\|\phi(n) A(n)\| \leq 1$. Dropping the argument $n$, 
\begin{align*}\|\phi^{k+1} {A}^{k+1} -\phi^{k+1} \bar{A}^{k+1} \| & = \| \phi(A-\bar{A})(\phi^kA^k)+ \phi \bar{A} (\phi^kA^k - \phi^k \bar{A}^k)\|
\\ & \leq  \|  \phi(A-\bar{A})\|  \|\phi^k A^k \| + \|\phi \bar{A}\|\|\phi^k A^k - \phi^k\bar{A}^k\|,\end{align*}
where the second line is by the triangle inequality and submultiplicativity of the matrix norm. Thus, the left-hand side increases by at most $\epsilon/K^2$ each time we increase $k$ by one. So by induction on $k$, the left-hand side is less than $\epsilon/K$ for all $k < K$.

Then 
\begin{align*} \|\textbf{c}( {A}(n),\phi(n)) - \textbf{c}( \bar{A}(n),\phi(n))\| &  \leq \sum_{k = 0}^{\infty} \phi(n)^k \|A(n)^k - \bar{A}(n)^k\| \|\mathbf{1}_n\|
\\ & = \sum_{k = 0}^{K-1} \phi(n)^k \|A(n)^k - \bar{A}(n)^k\| \|\mathbf{1}_n\| \\ & + \sum_{k=K}^{\infty} \phi(n)^k \|A(n)^k - \bar{A}(n)^k\| \|\mathbf{1}_n\|.
\end{align*}
The first term is less than $\epsilon \|\mathbf{1}_n\|$ because each summand is less than $\epsilon/K$. The second term is less than $2\epsilon\|\mathbf{1}_n\|$  by the bounds on  $ \left\|\sum_{k=K}^{\infty} \phi(n)^k \bar{A}(n)^k\right\|$ and $ \left\|\sum_{k=K}^{\infty} \phi(n)^k {A}(n)^k\right\|$. So the difference of Katz-Bonacich centralities is less than $3\epsilon \|\mathbf{1}_n\| = 3\epsilon \sqrt{n}.$
\end{proof}

\begin{lemma}\label{stochlem}
Suppose ${A}(n)$ is a sequence of random networks generated from a stochastic block model (with $m$ groups and fixed positive link probabilities) with non-vanishing spectral gap. Let $\epsilon > 0$. For $n$ sufficiently large, with probability at least $1-\epsilon$ we have
$$\frac{{v}_i(A(n))}{{v}_i(\bar{A}(n))} \in [1-\epsilon,1+\epsilon] \mbox{ and }\frac{{c_i}(A(n))}{{c_i}(\bar{A}(n))} \in [1-\epsilon,1+\epsilon]$$
for all $i$.
\end{lemma}
\begin{proof}
We first consider eigenvector centrality. Because the link probabilities do not depend on the size of the matrix, the eigenvalue $\overline{\lambda}_1$ and the maximum expected degree $\Delta$ are $O(n)$, so the sequence has large enough eigenvalues. By Theorem~\ref{thm1},
$$\|\textbf{v}(A(n)) - \textbf{v}(\bar{A}(n))\|_2 < \epsilon$$
with probability $1-\epsilon/2$ for each $n$ sufficiently large.

By a standard bound on the ratio between the norms $\|\cdot\|_1$ and $\|\cdot \|_2$, we have 
$$\|\textbf{v}(A(n)) - \textbf{v}(\bar{A}(n))\|_1 < \epsilon \sqrt{n}.$$

For each $i$,
$$v_i(A(n)) = {\lambda}_1^{-1} \sum_{j \in N_i} v_j(A(n))$$
and
$$v_i(\bar{A}(n)) = \overline{\lambda}_1^{-1} \sum_{j } \bar{A}(n)_{ij} v_j(\bar{A}(n)).$$
Applying the bound on $\|\textbf{v}(A(n)) - \textbf{v}(\bar{A}(n))\|_1$,
$$|v_i(A(n)) - {\lambda}_1^{-1} \sum_{j \in N_i} v_j(\bar{A}(n))| = |\lambda_1^{-1} \sum_{j \in N_i}v_j(A(n)) - {\lambda}_1^{-1} \sum_{j \in N_i} v_j(\bar{A}(n))|< {\lambda}_1^{-1}  \epsilon \sqrt{n}.$$

Now, we will use Chebyshev's theorem to bound the number of connections to each group uniformly in $i$. We consider the event $E_{ik}$ that the number $ \sum_{j \in N_i, \text{ group k}} A(n)_{ij}$ of agents in a given group $k$ who are neighbors of $i$ is within a factor of $\epsilon$ of its expected value $\sum_{j \text{ in group }k}\bar{A}(n)_{ij}$.
By Chebyshev's theorem, the probability of the complement of $E_{ik}$ vanishes at an exponential rate. Thus, for $n$ sufficiently large we can assume that $E_{ik}$ holds for all $i$ and $k$ with probability at least $1-\epsilon/2$.

When this event $E_{ik}$ occurs, we have
$$|\sum_{j \in N_i} v_j(\bar{A}(n))- \sum_{j} \bar{A}_{ij} v_j(\bar{A}(n))| < C \epsilon \sqrt{n}$$
for some constant $C \geq \sqrt{n} \max_j \bar{A}_{ij} v_j(\bar{A}(n))$ independent of $n$.

Thus for $n$ large, with probability at least $1-\epsilon$, we have
$$|v_i(A(n)) - {\lambda}_1^{-1} \sum_{j} \bar{A}_{ij} v_j(\bar{A}(n))| < {\lambda}_1^{-1} ( \epsilon \sqrt{n} + C \epsilon \sqrt{n})$$
 Combining this with the bound on $|\lambda_1 - \overline{\lambda}_1|$ from the proof of Theorem 1, we have
\begin{align*}|v_i(A(n))-v_i(\bar{A}(n))| & \leq |v_i(A(n)) - {\lambda}_1^{-1} \sum_{j} \bar{A}_{ij} v_j(\bar{A}(n))|+|\overline{\lambda}_1^{-1} \sum_{j} \bar{A}_{ij} v_j(\bar{A}(n))-{\lambda}_1^{-1} \sum_{j} \bar{A}_{ij} v_j(\bar{A}(n))|
\\ & = |v_i(A(n)) - {\lambda}_1^{-1} \sum_{j} \bar{A}_{ij} v_j(\bar{A}(n))|+|\lambda_1-\overline{\lambda}_1| \cdot \lambda_1^{-1} \cdot | \overline{\lambda_1}^{-1} \sum_{j} \bar{A}_{ij} v_j(\bar{A}(n))|
\\ & = {\lambda}_1^{-1} ( \epsilon \sqrt{n} + C \epsilon \sqrt{n} )+ \epsilon v_i(\bar{A}(n)).\end{align*}
Because $\lambda_1^{-1}$ grows at a linear rate while $v_i(\bar{A}(n))$ is $O(n^{-1/2})$, the desired bound follows.

The proof for Katz-Bonacich centrality follows by the same argument, up to the difference in normalizations. We now use the expressions
$$c_i(A(n), \phi(n)) = 1 + \phi(n) \sum_{j\in N_i} c_j(A(n), \phi(n))$$
and
$$c_i(\bar{A}(n), \phi(n)) = 1 + \phi(n) \sum_{j}\bar{A}(n)_{ij} c_j(\bar{A}(n), \phi(n))$$
from Definition~\ref{katzbondef} instead of the recursive expressions of eigenvector centrality.
\end{proof} 

\begin{proof}[Proof of Proposition~\ref{twoeig}]

Lemma~\ref{stochlem} tells us that for a sequence of random matrices generated as in the proposition statement with population size $N_i$, each of the entries of the first eigenvector of each matrix $A$ is close to the corresponding entry of the first eigenvector of $\bar{A}$ with probability close to one.

\begin{figure}
\center{\includegraphics[scale=.55]{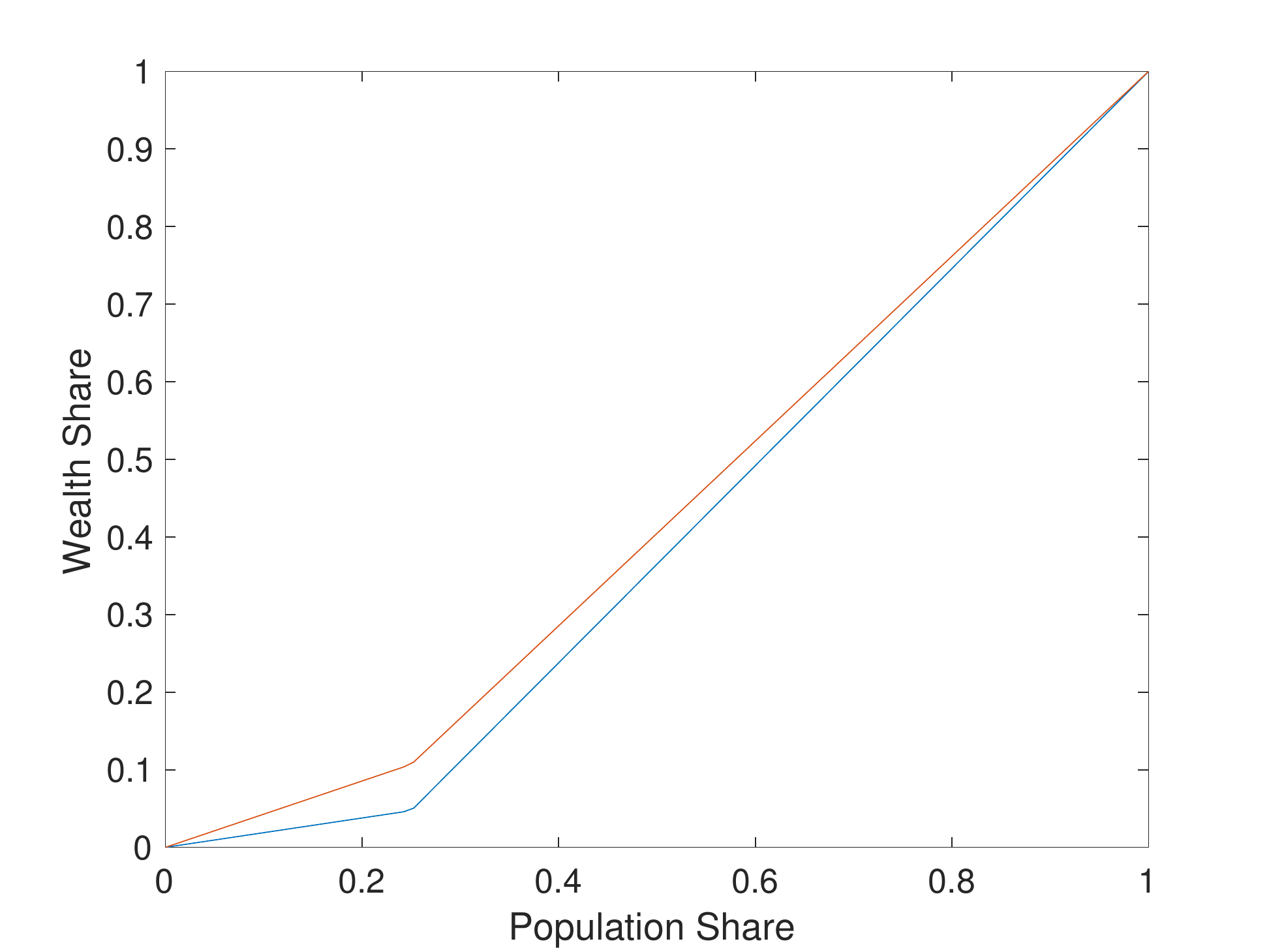}}
\caption{Lorenz curves for eigenvector centralities of two groups with $p_s = .5, p_d = .05$, $s_1 = .75$ in blue, and $p_s' = .4, p_d'=.1,s_1'=.75$ in red. The red distribution Lorenz dominates the blue distribution.}\label{fig:groupsize}
\end{figure}

We first show (i). By definition, the eigenvector centralities satisfy
$$v_i(\bar{A}(n)) = \overline{\lambda}_1^{-1} (s_ln (p_s-p_d) v_i(\bar{A}(n)) +p_d \sum_{j=1}^n v_j(\bar{A}(n))),$$
where $s_l$ is the size of the group containing $i$. For any $i'$ in a group of size $s_{l'}$, we have
$$ \frac{v_i(\bar{A}(n)) }{v_{i'}(\bar{A}(n)) }=\frac{1-\overline{\lambda}_1^{-1}s_{l'}n (p_s-p_d) }{1-\overline{\lambda}_1^{-1}s_{l}n (p_s-p_d) }.$$
Thus,
$$\frac{\partial}{\partial p_s}\left( \frac{v_i(\bar{A}(n)) }{v_{i'}(\bar{A}(n)) }\right)=-\frac{\partial}{\partial p_d}\left( \frac{v_i(\bar{A}(n)) }{v_{i'}(\bar{A}(n)) }\right)=\frac{\overline{\lambda}_1^{-1}(s_{l}-s_{l'})n}{(1-\overline{\lambda}_1^{-1}s_{l}n (p_s-p_d))^2}.$$
The right-hand side is positive if and only if $s_l > s_{l'}$.

Now, we can order the agents $1,\hdots,n$ so that their group sizes are increasing. Increasing $p_s$ and decreasing $p_d$ will increase $\frac{v_i(\bar{A}(n))}{v_{i'}(\bar{A}(n))}$ for any $i > i'$. Thus, this change will decrease
$$\frac{\sum_{i=1}^k v_i(\bar{A}(n))}{\sum_{i=1}^n v_i(\bar{A}(n)) }=\left(\sum_{i=1}^n \frac{1}{\sum_{i'=1}^k\frac{v_{i'}(\bar{A}(n))}{v_{i}(\bar{A}(n))}}\right)^{-1}.$$
for each $1 \leq k < n$. This proves the Lorenz dominance result.

To prove (ii), we compute the eigenvector centrality when there are two groups from the first eigenvector of the 2-by-2 matrix $$\begin{pmatrix} s_1N p_s & (1-s_1)Np_d \\ s_1Np_d & (1-s_1)Np_s \end{pmatrix},$$
which replaces each group with a representative agent.

The first eigenvector of $$\begin{pmatrix} s_1N p_s & (1-s_1)Np_d \\ s_1Np_d & (1-s_1)Np_s \end{pmatrix}$$
is equal to 
$$\begin{pmatrix} 1 \\ \frac{(1-2s_1)p_s + \sqrt{(1-2s_1)^2 p_s^2 + 4s_1(1-s_1)p_d^2} }{(2-2s_1)p_d}\end{pmatrix},$$
up to rescaling.

So the first eigenvector of $\bar{A}$ has first $s_1N$ entries equal to $1$ and final $(1-s_1)N$ entries equal to $$\frac{(1-2s_1)p_s + \sqrt{(1-2s_1)^2 p_s^2 + 4s_1(1-s_1)p_d^2} }{(2-2s_1)p_d},$$
up to rescaling.

Let \begin{align*}T(s) & = sN + \frac{(1-2s_1)p_s + \sqrt{(1-2s_1)^2 p_s^2 + 4s_1(1-s_1)p_d^2} }{2p_d} \cdot N \end{align*}
be the total of the entries of the first eigenvector.

Differentiating this shows that $\frac{\partial T}{\partial s}$ is non-negative if and only if
$$s_1 \leq \frac12 + \frac{p_d}{2\sqrt{2p_d^2+2p_sp_d}}.$$

We conclude that the first eigenvector corresponding to $(s,p_s,p_d)$ Lorenz dominates the first eigenvector corresponding to $(s_1',p_s,p_d)$ if $$s'_1 \leq  \frac12 + \frac{p_d}{2\sqrt{2p_d^2+2p_sp_d}}$$
and $s_1 \leq s'_1$. \end{proof}

\begin{figure}

\center{\includegraphics[scale=.55]{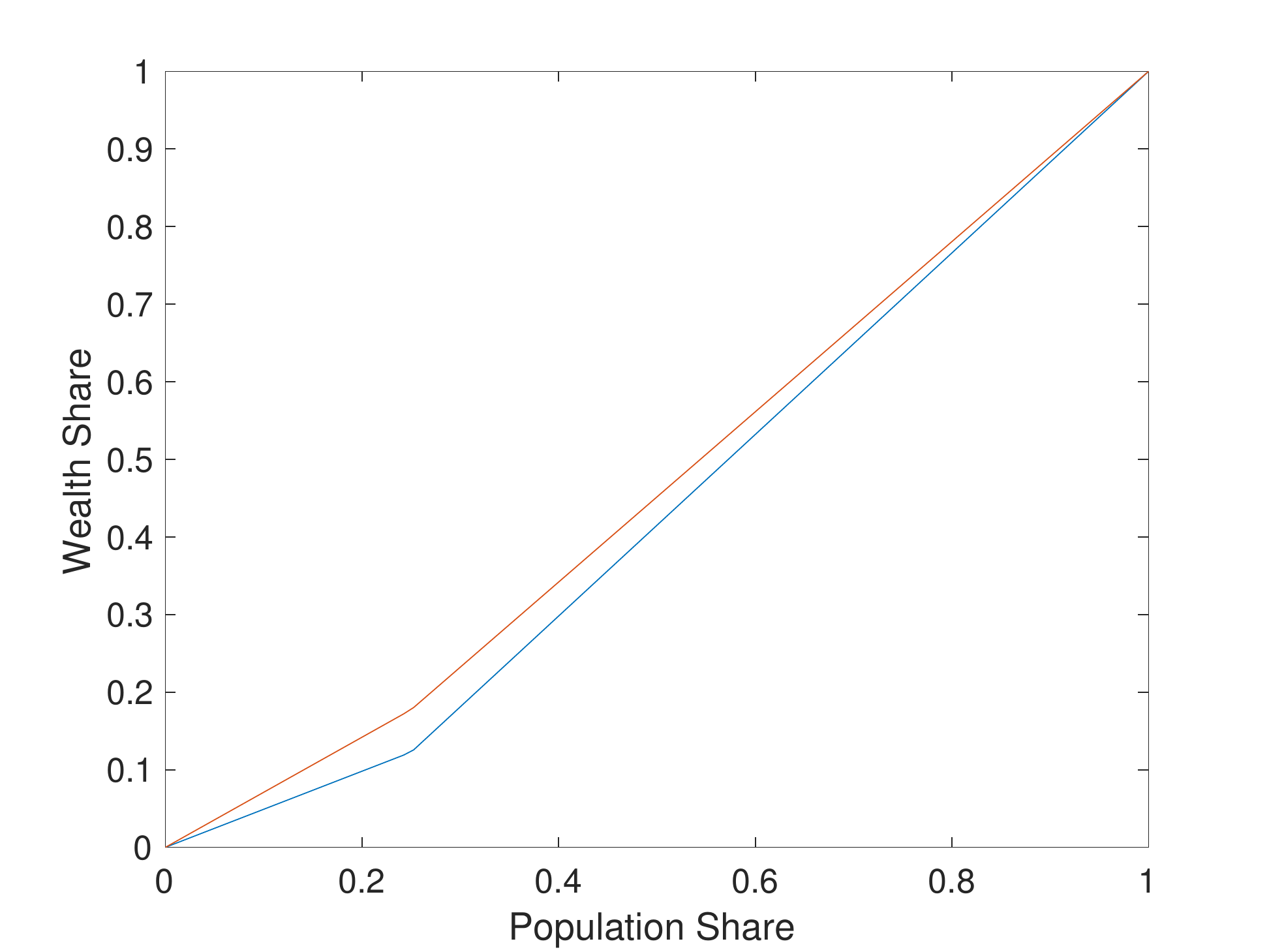}}
\caption{Lorenz curves for Katz-Bonacich centralities of two groups with $100$ agents and parameter values $p_s = .5, p_d = .05$, $s_1 = .75$, $\phi=.02$ in blue, and $p_s' = .4, p_d'=.1,s_1'=.75$, $\phi=.02$ in red. The red distribution Lorenz dominates the blue distribution.}\label{fig:groupsize}
\end{figure}

\begin{proof}[Proof of Proposition~\ref{twobon}]

The proof follows the same argument as the proof of Proposition~\ref{twoeig}. Definition \ref{katzbondef} now gives
$$\frac{c_i(\bar{A},\phi)}{c_{i'}(\bar{A},\phi)}  = \frac{1-s_{l'}(p_s-p_d)\phi n}{1-s_{l}(p_s-p_d)\phi n}.$$
Thus, we compute the elasticities
$$\frac{\partial }{\partial p_s}\left( \frac{c_i(\bar{A},\phi)}{c_j(\bar{A},\phi)} \right)= -\frac{\partial }{\partial p_d} \left( \frac{c_i(\bar{A},\phi)}{c_j(\bar{A},\phi)}
\right) =\frac{(s_l-s_{l'})(p_s-p_d)\phi n}{(1-s_l(p_s-p_d)\phi n)^2}.$$
Just as for eigenvector centrality, this has the same sign as $s_l-s_l'$. The remainder of the proof is the same.

\end{proof}

\begin{proof}[Proof of Proposition~\ref{same}]
We want to show that $$\frac{\partial}{\partial p_s}  c_i(\bar{A},\phi) > \frac{\partial}{\partial p_s}  c_{i'}(\bar{A},\phi)$$when $i$ is in a larger group than $i'$. By Lemma~\ref{pathformula}, we have 
$$\frac{\partial}{\partial p_s}  c_i(\bar{A},\phi) = \frac{1}{p_{s}} \sum_{k=0}^{\infty} \phi^k \sum_{\gamma^k \text{ beginning at }i}   f(\gamma^k) \prod_{j=0}^{k-1} p_{i_ji_{j+1}},$$
where the inner sum is over walks of length $k$ beginning at agent $i$ and $f(\gamma^k)$ is the number of edges of $\gamma^k$ from a group to itself.

We claim that $$\sum_{\gamma^k \text{ beginning at }i}   f(\gamma^k) \prod_{j=0}^{k-1} p_{i_ji_{j+1}}> \sum_{\gamma^k \text{ beginning at }i'}   f(\gamma^k) \prod_{j=0}^{k-1} p_{i_ji_{j+1}}$$
for each $k$.

For length $1$ walks, note that the number of walks from $i$ to an agent in another group is independent of $i$'s group, and each of these walks has $f(\gamma^1)=0$. The number of walks from $i$ to an agent in the same group is increasing in the size of $i$'s group, and each of these walks has $f(\gamma^1)=1$. So we can choose an injection $\phi_1$ from the set of length $1$ walks beginning at $i'$ to the set of length $1$ walks beginning at $i$ preserving $f(\gamma^1)$ and edge weights and sending each walk $\gamma^1$ to a walk $\phi_1(\gamma^1)$ such that the group $\phi_1(\gamma^1)$ ends in is at least as large as the group that $\gamma^1$ ends in.

Given a map $\phi_{k-1}$ satisfying these properties, we can construct $\phi_k$ satisfying the same properties. We send $\gamma^k$ to a walk defined by using $\phi_{k-1}$ to determine the first $k$ agents and then defining the last agent as in the previous paragraph. This proves the claim.
\end{proof}

\begin{proof}[Proof of Corollary~\ref{pscor}]
By Proposition~\ref{same},
$$c_i({\bar{A}'(n)},\phi(n)) - c_i(\bar{A}(n),\phi(n)) > c_j(\bar{A}'(n), \phi(n))-c_j(\bar{A}(n),\phi(n))$$
for all $n$. We must check that Theorem~\ref{thm2} applies.

The sequence has large enough eigenvalues because $\overline{\lambda}_1(n)$ is $O(n)$ and $\Delta$ is linear in $n$. The non-zero eigenvalues of $\bar{A}(n)$ are proportional to the eigenvalues of the matrix
$$\begin{pmatrix} s_1p_s & s_1 p_d & \hdots  & s_1p_d
\\ s_2p_d & s_2p_s & \hdots & s_2p_d \\ \vdots & \vdots & \ddots & \vdots \\ s_m p_d & s_m p_d & \hdots & s_m p_s  \end{pmatrix},$$
where $s_1,s_2,...,s_m$ are the population shares in each group. By the Perron-Frobenius theorem the first and second eigenvalues of this matrix are distinct, so the sequence has non-vanishing spectral gap.

Choose $\epsilon$ small enough so that 
$$c_i({\bar{A}'(n)},\phi(n)) - c_i(\bar{A}(n),\phi(n)) - 2\epsilon > c_j(\bar{A}'(n), \phi(n))-c_j(\bar{A}(n),\phi(n)) + 2\epsilon.$$
With probability approaching $1$ as $n \rightarrow \infty$, $$c_i({\bar{A}'(n)},\phi(n)) -c_i(A'(n),\phi(n)) <\epsilon$$
and similarly for the other three centrality terms. So the desired inequality holds with probability approaching $1$ as $n\rightarrow \infty$.
\end{proof}

\begin{proof}[Proof of Proposition~\ref{dif}]
(i) Let $d_i(\bar{A})$ denote the degree of agent $i$. As $\phi \rightarrow 0$, the Katz-Bonacich centrality $$c_i(\bar{A},\phi) = 1 + \phi d_i(\bar{A}) + O(\phi^2)$$
is equal to positive affine transformation of the degree centrality of $\phi$, plus terms of order $\phi^2$. As $p_d$ increases, the change in degree $\frac{\partial }{\partial p_d} d_i(\bar{A})$ is equal to the total number of agents $n$ minus the size of agent $i$'s group. In particular, this derivative is decreasing in the group size.

(ii) Suppose $i$'s group is larger than $j$'s group.

By Lemma~\ref{pathformula}, the derivative of the Katz-Bonacich centrality of agent $i$ is
\begin{equation}\label{eq:bonderiv}\frac{\partial }{\partial p_{d}}c_i(\bar{A},\phi) = \frac{1}{p_d} \sum_{k=0}^{\infty} \phi^k  \sum_{\gamma^k \text{ beginning at }i}   f(\gamma^k) \prod_{l=0}^{k-1} p_{i_li_{l+1}}\end{equation}
where the sum is over walks $\gamma^k$ of length $k$ starting at agent $i$, the indices $i_0,i_1,...,i_k$ are the groups that the walks pass through, and $f(\gamma^k)$ is the number of times $\gamma^k$ switches groups.

Consider the Markov chain with states corresponding to the groups in our network and transition probabilities from state $s'$ to state $s''$ given by $p_{s's''}/\sum_{s=1}^m p_{s's}.$ The probability that the Markov chain is in state $s$ at time $k$ given an initial state is equal to the probability that a random walk beginning in the group corresponding to the initial state of length $k$ ends at $s$, where each walk is chosen with probability proportional to $ \prod_{l=0}^{k-1} p_{i_li_{l+1}}$.

The Markov chain converges to a stationary distribution. Therefore as $k \rightarrow \infty$, the expected number of times the Markov chain has switched states by time $k$ approaches $ck$ for some constant $c$ independent of the initial state. So as $k$ grows large, $$ \sum_{\gamma^k: i \text{ in group }i_0 }   f(\gamma^k) \prod_{l=0}^{k-1} p_{i_li_{l+1}} \rightarrow c \sum_{\gamma^k: i \text{ in group }i_0 } \prod_{l=0}^{k-1} p_{i_li_{l+1}} .$$
Because $i$'s group is larger than $j$'s group,
$$\frac{\sum_{\gamma^k: i \text{ in group }i_0 } \prod_{l=0}^{k-1} p_{i_li_{l+1}} }{\sum_{\gamma^k: j \text{ in group }i_0 } \prod_{l=0}^{k-1} p_{i_li_{l+1}} }$$
is then bounded below by a constant greater than $1$.

As $\phi$ approaches $\overline{\lambda}^{-1}$, the ratio on the right-hand side of equation \ref{eq:bonderiv} diverges to $\infty$. Moreover, we showed in the previous paragraph that $$\frac{ \sum_{\gamma^k: i \text{ in group }i_0 }   f(\gamma^k) \prod_{l=0}^{k-1} p_{i_li_{l+1}}}{\sum_{\gamma^k: j \text{ in group }i_0 }   f(\gamma^k) \prod_{l=0}^{k-1} p_{i_li_{l+1}}}$$
is bounded below by a constant greater than $1$ for all but finitely many $k$. This implies that for all but finitely many values of $k$, the terms in equation \ref{eq:bonderiv} corresponding to length $k$ paths are at least this constant times the same terms when the starting agent $i$ is replaced by $j$. So summing over $k$, we can conclude that
$$ \frac{\partial }{\partial p_{d}}c_i(\bar{A},\phi) > \frac{\partial }{\partial p_{d}}c_j(\bar{A},\phi)$$
for $\phi$ large enough.

\end{proof}
\begin{proof}[Proof of Corollary~\ref{pdcor}]
Fix some $n_0$. By Proposition~\ref{dif}, there exists $\underline{\phi}$ such that
$$c_i({A'}(n_0),\phi(n_0)) - c_i({A}(n_0),\phi(n_0)) > c_j({A}'(n_0), \phi(n_0))-c_j({A}(n_0),\phi(n_0)) $$
whenever $0 < \phi < \underline{\phi}$ and $i$'s group is larger than $j$'s group. 

The Katz-Bonacich centralities of $\bar{A}(n_0)$ are a positive affine transformation of the Katz-Bonacich centralities of the matrix
$$\begin{pmatrix} s_1p_s & s_1 p_d & \hdots  & s_1p_d
\\ s_2p_d & s_2p_s & \hdots & s_2p_d \\ \vdots & \vdots & \ddots & \vdots \\ s_m p_d & s_m p_d & \hdots & s_m p_s  \end{pmatrix},$$
where $s_i$ are the population shares of each group. This is because the centralities are proportional except the constant term corresponding to length zero walks. The same statement holds for $\bar{A}'(n_0)$ with the same rescaling.

The upshot is that the inequality comparing centralities does not depend on the choice of $n_0$. So it follows that for any $n$ such that $\phi(n)/\overline{\lambda}_1(n) < \underline{\phi}/\overline{\lambda}_1(n_0),$ we have
$$c_i({A'}(n),\phi(n)) - c_i({A}(n),\phi(n)) > c_j({A}'(n), \phi(n))-c_j({A}(n),\phi(n)).$$

The remainder of the proof of part (i) proceeds as in the proof of Corollary~\ref{pscor}, and the analogous argument shows part (ii).
\end{proof}

\begin{proof}[Proof of Proposition~\ref{exam}]
Suppose we have three groups of equal size $K = n/3$ such that $p_{11} = 1,$ $p_{12} = p_{21} = \delta$, and all other link probabilities are $0$. We will show that for $\phi$ sufficiently large and $\delta$ sufficiently small, $$\frac{\partial }{\partial p_{23}} c_1(\bar{A},\phi) >\frac{\partial }{\partial p_{23}} c_{n/3+1}(\bar{A},\phi).$$
Note that agent $1$ is in group $1$ while agent $n/3+1$ is in group $2$.

To check this identity, note that $$(I - \phi P)^{-1}  \approx \begin{pmatrix} \frac{1}{1-\phi K } & \frac{\phi K \delta}{1-\phi K } &  0 \\  \frac{\phi K \delta}{1-\phi K } & 1 & 0 \\ 0 & 0 & 1 \end{pmatrix},$$
where $P$ is the $3 \times 3$ matrix with entries $P_{ij}=Kp_{ij}$ equal to the number of links from a given agent in group $i$ to any agent in group $j$. The approximation is dropping terms of order $\delta^2$. The entries $(I-\phi A)^{-1}_{kl}$ are equal to the entries $(I-\phi P)^{-1}_{ij}$ corresponding to the groups of $k$ and $l$.

Using Lemma~\ref{matrixformula}, we compute
$$\frac{\partial }{\partial p_{23}} c_1(\bar{A},\phi)\approx \frac{\phi^2 K^2 \delta}{1-\phi K } \text{ and }\frac{\partial }{\partial p_{23}} c_{n/3+1}(\bar{A},\phi)  \approx \phi K .$$
Fixing $\delta$ sufficiently small and taking $\phi  \rightarrow 1/K$ (keeping $\phi$ feasible), we find the first expression is larger than the second for all $\phi$ sufficiently large.
\end{proof}

\begin{proof}[Proof of Corollary~\ref{excor}]
By Proposition~\ref{exam}, we can choose $\bar{A}$ and $\bar{A}'$ with $n_0$ agents each such that $p_{ij}$ is larger in $\bar{A}'$ than $\bar{A}$ and other link probabilities agree, groups $i,j$ and $k$ and a constant $\underline{\phi} < \overline{\lambda}_1$ such that
$$c_k(\bar{A}',\phi) - c_k(\bar{A},\phi)>c_i(\bar{A}',\phi) - c_i(\bar{A},\phi)$$
whenever $\underline{\phi} < \phi < \overline{\lambda}_1$.
 
Define sequences of stochastic block networks by letting all relative group sizes and link formation-probabilities be the same as in $\bar{A}$ and $\bar{A}',$ respectively. Let $\underline{\phi}(n) = \frac{n}{n_0}\underline{\phi}$ for each $n$. Then arguing as in the proof of Corollary~\ref{pdcor}, we find that 
$$c_k(\bar{A}'(n),\phi(n)) - c_k(\bar{A}(n),\phi(n))>c_i(\bar{A}'(n),\phi(n)) - c_i(\bar{A}(n),\phi(n))$$
for each $n$.

The proof concludes with the same arguments as the proof of Corollary~\ref{pscor}.
\end{proof}

\begin{proof}[Proof of Claim~\ref{gridclaim}]
By equation (\ref{leveq}), we must check that
$$(\min_{i,j}\bar{A}(n)_{ij} )\sum_{i} \sum_j \bar{A}(n)_{ij}\geq \Delta \log n,$$
where $n = (k+1)^2$.

We first note that
$$\Delta = \max_i \sum_j \bar{A}(n)_{ij} \leq 4 \min_i \sum_j \bar{A}(n)_{ij}.$$
This is because the maximum is obtained at the node(s) at the center of the grid while the minimum is obtained by the corner nodes. Dividing the grid into four quadrants vertically and horizontally, the expected number of links from a corner node to the nearest quadrant is equal to the expected number number of links from a center node to any quadrant.

Since $\min_{i,j}\bar{A}(n)_{ij} \geq (\sqrt{2} k)^{-\rho}$, we conclude that $$\min_{i,j}\bar{A}(n)_{ij} \sum_{i} \sum_j \bar{A}(n)_{ij} \geq (\sqrt{2} k)^{-\rho} \cdot (4 \Delta n).$$
The right-hand side grows faster than $\Delta \log n$.
\end{proof}

\begin{proof}[Proof of Lemma~\ref{layers}]
See for example Theorem 13.12 of \cite*{Laub05}.
\end{proof}

\begin{proof}[Proof of Proposition~\ref{layerscor}]
By Theorem~\ref{thm1}, we have
$$\|\textbf{v}(A^{(i)}(n_i) )- \textbf{v}(\bar{A}^{(i)}(n_i))\|_2 < {\epsilon}/2$$
for $n_i$ sufficiently large and $i=1,2$. By the triangle inequality 
$$  \| \textbf{v}(\bar{A}^{(1)}(n_1))\otimes \textbf{v}(\bar{A}^{(2)}(n_2))-\textbf{v}(A^{(1)}(n_1))\otimes \textbf{v}(A^{(2)}(n_2))\|_2$$ is less than or equal to  \begin{align*} \| \textbf{v}(\bar{A}^{(1)}(n_1))\otimes \textbf{v}({A}^{(2)}(n_2))-\textbf{v}(A^{(1)}(n_1))\otimes \textbf{v}(A^{(2)}(n_2))\|_2  + \\  \| \textbf{v}(\bar{A}^{(1)}(n_1))\otimes \textbf{v}(\bar{A}^{(2)}(n_2))-\textbf{v}(\bar{A}^{(1)}(n_1))\otimes \textbf{v}(A^{(2)}(n_2))\|_2.\end{align*}
Because $\textbf{v}$ is normalized to have Euclidean norm one, each of these two terms is less than $\epsilon/2$.

Because $\overline{\lambda}_1,\Delta$ and $\log(n)$ are all multiplicative in the Kronecker product, $\bar{A}(n_1n_2)$ has non-vanishing spectral gap and large enough eigenvalues. So by Theorem~\ref{thm1},
$$\|\textbf{v}(A(n_1n_2) )- \textbf{v}(\bar{A}(n_1n_2))\|_2 < {\epsilon}$$
for $n_1n_2$ sufficiently large.

So we have\begin{align*}\|\textbf{v}(A(n_1n_2))-\textbf{v}(A^{(1)}(n_1))\otimes  \textbf{v}(A^{(2)}(n_2))\|_2 & \leq \|\textbf{v}(\bar{A}(n_1n_2))-\textbf{v}(\bar{A}^{(1)}(n_1))\otimes  \textbf{v}(\bar{A}^{(2)}(n_2))\|_2 \\ & + \| \textbf{v}(\bar{A}^{(1)}(n_1))\otimes \textbf{v}(\bar{A}^{(2)}(n_2))-\textbf{v}(A^{(1)}(n_1))\otimes \textbf{v}(A^{(2)}(n_2))\|_2 
\\ & + \|\textbf{v}(\bar{A}(n_1n_2))-\textbf{v}(A(n_1n_2))\|_2.\end{align*}
The first term is zero by Lemma~\ref{layers}. The second and third term are each at most $\epsilon$ by the preceding paragraphs.
\end{proof}

%%END OF PROOFS

\section{Extensions}\label{ext}

\subsection{Weighted Edges}\label{weights}

In this subsection, we relax the assumption that $A_{ij} \in \{0,1\}$ to allow non-integral edge weights $A_{ij} \geq 0$. The centrality measures $\textbf{v}(A)$ and $\textbf{c}(A)$ are defined as before.

Suppose that for each pair of groups $i$ and $j$, the edge weights between agents in group $i$ and group $j$ are independent uniform random variables on some interval $[l_{ij}(n),u_{ij}(n)]$. Then Theorem~\ref{thm1} and Theorem~\ref{thm2} hold as before with $\bar{A}(n) = \mathbb{E}[A(n)]$.

\begin{proof}
First note that we can approximate uniform random variables on some non-negative interval $[l,u]$ by linear combinations of i.i.d. Bernoulli random variables with positive coefficients. We discuss the case $l=0$ and $u=1$, and the argument extends easily.

Consider the weighted sum with coefficients $2^{-k}$ for $k = 1,...,m$ of i.i.d. Bernoulli random variables equal to $1$ with probability $\frac12$. As $m \rightarrow \infty$, this sum converges in probability to the uniform random variable on $[0,1]$.

In the proof of Theorem 1 of \cite*{Chung11}, the authors express $A$ as a sum of independent random matrices $X_{ij}$, where $X_{ij}$ is the product of a Bernoulli random variable equal to one with probability $p_{ij}$ and the matrix with entries $(i,j)$ and $(j,i)$ equal to one and all other entries zero. We instead express $A$ as a sum of independent variables $X_{ij, k}$ corresponding to the terms in the sum in the previous paragraph. Because all relevant quantities are continuous and our previous arguments do not depend on the number of random variables, our proof proceeds as before when we take $m$ sufficiently large.
\end{proof}

\subsection{Clustering}\label{clustering}

When edges have non-integral weights as in Appendix \ref{weights}, it is straightforward to allow clustering in the network. We present a variation of our model in which weights between $i$ and $j$ are likely to be higher when both agents are connected to some other agent $k$. This tendency is produced by allowing triangles to form in addition to pairwise connections.

As before, for each pair of agents $k$ in group $i$ and $l$ in group $j$ we draw an independent Bernoulli random variable equal to one with probability $p_{ij}$. In addition, for each triplet of agents $k,k'$ and $k''$ in groups $i,i'$ and $i''$ we draw an independent Bernoulli random variable equal to one with probability $p_{ii'i''}$. The entry $A_{ij}$ is equal to the sum of these random variables over the pair $i,j$ and all triplets including $i$ and $j$.

So if the direct link between $i$ and $j$ does not form, $A_{ij}$ is equal to the number of triplets containing $i$ and $j$ which form. If the direct link between $i$ and $j$ does form, $A_{ij}$ is equal to one plus the number of triplets containing $i$ and $j$ which form.

Then Theorem~\ref{thm1} and Theorem~\ref{thm2} still hold with $\bar{A}(n) = \mathbb{E}[A(n)]$. The proof is a straightforward modification of the proof in Section~\ref{weights}.

Alternately, we can allow the weights on each link and each triplet to be uniform random variables. Suppose for each pair of agents $k$ in group $i$ and $l$ in group $j$ we draw a uniform random variable on $[l_{ij}(n),u_{ij}(n)]$. In addition, for each triplet of agents $k,k'$ and $k''$ in groups $i,i'$ and $i''$ we draw a uniform random variable on $[l_{ii'i''}(n),u_{ii'i''}(n)]$. If entry $A_{ij}$ is again equal to the sum of these random variables over the pair $i,j$ and all triplets including $i$ and $j$, then Theorems~\ref{thm1} and~\ref{thm2} continue to hold.

Finally, the discussion above includes only direct links and triangles. The results extend similarly to allow other subgraphs, such as cliques with more than three agents.

\section{Comparative Statics of Katz-Bonacich Centrality}\label{appendixstatics}

We give two formulas for the derivatives of Katz-Bonacich centrality in stochastic block models. Our first formula expresses the derivative of $c_l(A,\phi)$ as a weighted count of the number of walks starting at agent $l$. Let $\gamma^k$ denote a walk of length $k$, let $i_0,i_1,...,i_k$ be the groups the walk passes through, and let $f(\gamma^k)$ be the number of times $\gamma^k$ passes between groups $i$ and $i'$.

\begin{lemma}\label{pathformula}
The derivative of the Katz-Bonacich centrality of agent $l$ is
$$\frac{\partial }{\partial p_{ii'}}c_l(A,\phi) = \frac{1}{p_{ii'}} \sum_{k=0}^{\infty} \phi^k \sum_{\gamma^k\text{ beginning at }l}   f(\gamma^k) \prod_{j=0}^{k-1} p_{i_ji_{j+1}}.$$
This is equal to the number of walks with each counted with multiplicity given by the probability of the walk forming times the number of times the walk passes between groups $i$ and $i'$, each discounted by the walk length.
\end{lemma}

\begin{proof}[Proof of Lemma~\ref{pathformula}]
We have $$c_l(A,\phi) = \sum_{k=0}^{\infty}\phi^k  \sum_{\gamma^k\text{ beginning at }l }   \left(\prod_{j=0}^{k-1} p_{i_ji_{j+1}}\right).$$
The result follows from computing the derivative.
\end{proof}

Our second formula expresses the derivative of $c_k(A,\phi)$ explicitly in terms of the link probabilities $p_{ij}$. Denote the number of agents in each group $i$ by $s_in$. Let $P$, which can be thought of as the adjacency matrix of the network with each group replaced by a representative agent, be the $m \times m$ matrix with entries $P_{ij} = s_jnp_{ij}$.

\begin{lemma}\label{matrixformula}
The derivative of the Katz-Bonacich centrality of agent $k$ in group $l$ in terms of $p_{ij}$ is
$$\frac{\partial }{\partial p_{ij}} c_k(A,\phi)
 = \phi n\left( \sum_{l'=1}^m s_j(I-\phi P)^{-1}_{li} (I-\phi P)^{-1}_{jl'} +s_i (I-\phi P)^{-1}_{lj} (I-\phi P)^{-1}_{il'}\right)
.$$
\end{lemma}

\begin{proof}[Proof of Lemma~\ref{matrixformula}]
We have
\begin{align*} \frac{\partial}{\partial p_{ij}}  c_k(A,\phi)& = \sum_{t = 0}^{\infty} \frac{\partial (\phi^tA^t \mathbf{1})_k}{\partial p_{ij}}
\\ & = \sum_{t = 0}^{\infty} \sum_{t' = 0}^{\infty} ((\phi A)^{t} \frac{\partial  (\phi A)}{\partial p_{ij}} (\phi A)^{t'} \mathbf{1})_k.
\end{align*}

For any $l$ and $l'$, $(I-\phi P)_{ll'}$ is equal to the discounted number of walks which begin at some fixed agent in group $l$ and end at any agent in group $l'$. On the other hand, the entries of $\sum_{t=0}^{\infty} (\phi A)^t$ correspond to the discounted numbers of walks which begin at some fixed agent in group $l$ and end at some fixed agent in group $l'$.

The derivative $\frac{\partial  (\phi A)}{\partial p_{ij}}$ of $A$ as we vary the link probability between groups $i$ and $j$ has entries equal to $\phi$ when one index corresponds to an agent in group $i$ and the other to an agent in group $j$ and all other entries zero. So we compute that the double summation above is equal to
$$\phi n \left( \sum_{l'=1}^m s_j(I-\phi P)^{-1}_{li} (I-\phi P)^{-1}_{jl'} +s_i (I-\phi P)^{-1}_{lj} (I-\phi P)^{-1}_{il'}\right),$$
as desired. The two terms correspond to walks passing through an edge from group $i$ to group $j$ and walks passing in the opposite direction.\end{proof}

\bibliographystyle{ecta}

% List ALL references in your references, not just the ones cited in the text.
%This scheme automatically alphabetizes the Bibliography.
%\bibliographystyle{acm}

\bibliography{Distributions}

\end{document}